\documentclass[12pt]{iopart}
\expandafter\let\csname equation*\endcsname\relax
\expandafter\let\csname endequation*\endcsname\relax

\usepackage{
    amsfonts,
    amsmath,
    amssymb,
    amsthm,
    bm,
    braket,
    dsfont,
    graphicx,
    latexsym,
    mathpazo,
    mathrsfs, 
    subcaption 
}

\usepackage{hyperref}
\usepackage{amsthm}
\usepackage[autostyle]{csquotes}
\usepackage{bm}

\theoremstyle{definition}

\usepackage{amssymb,amscd}
\usepackage{todonotes}
\usepackage[compatibility=false]{caption}
\usepackage{graphicx}
\usepackage{braket}
\usepackage{dsfont}
\usepackage{color}
\usepackage{tikz}
\usetikzlibrary{arrows,positioning}
\usetikzlibrary{shapes.geometric}
\usetikzlibrary{patterns}
\usepackage{mathrsfs}
\usepackage{algorithm}
\usepackage{algpseudocode}
\usepackage{mathtools}
\usepackage{tcolorbox}
\usepackage[utf8]{inputenc}
\usepackage{enumerate}
\usepackage[english]{babel}
\usepackage{lipsum} 

\algrenewcommand{\Return}{\State\algorithmicreturn~}
\newtheorem{prop}{Proposition}

\newcommand\norm[1]{\left\lVert#1\right\rVert}

\algdef{SE}[DOWHILE]{Do}{doWhile}{\algorithmicdo}[1]{\algorithmicwhile\ #1}%
\algnewcommand{\LineComment}[1]{\State \(\triangleright\) #1}

\makeatletter
\renewcommand*\env@matrix[1][*\c@MaxMatrixCols c]{%
	\hskip -\arraycolsep
	\let\@ifnextchar\new@ifnextchar
	\array{#1}}
\makeatother

 \algnewcommand{\Initialize}[1]{%
  \State \textbf{Initialize:}
  \Statex \hspace*{\algorithmicindent}\parbox[t]{.8\linewidth}{\raggedright #1}
}
 
 \usepackage{xparse}

\NewDocumentCommand{\INTERVALINNARDS}{ m m }{
    #1 {,} #2
}
\NewDocumentCommand{\interval}{ s m >{\SplitArgument{1}{,}}m m o }{
    \IfBooleanTF{#1}{
        \left#2 \INTERVALINNARDS #3 \right#4
    }{
        \IfValueTF{#5}{
            #5{#2} \INTERVALINNARDS #3 #5{#4}
        }{
            #2 \INTERVALINNARDS #3 #4
        }
    }
}
\begin{document}
\title[Continuous-variable ramp quantum  secret sharing with Gaussian states and operations]{Continuous-variable ramp quantum  secret sharing \\
with Gaussian states and operations}
\author{%
Masoud Habibidavijani $^{1}$
and Barry C. Sanders $^{1,2}$
}
\address{$^1$%
Institute for Quantum Science and Technology,
University of Calgary,
Alberta T2N 1N4, Canada%
}
\address{$^2$%
Program in Quantum Information Science, Canadian Institute for Advanced Research,Toronto, Ontario M5G 1Z8, Canada
}
\begin{abstract}								
Our aim is to formulate continuous-variable quantum secret sharing  as a continuous-variable ramp quantum  secret sharing  protocol, provide a certification procedure for it and explain the criteria for the certification. Here we  introduce a technique for certifying  continuous-variable ramp quantum secret-sharing schemes  in the framework of quantum interactive-proof systems. We devise pseudocodes in order to represent the sequence of steps taken to solve the certification problem. Furthermore, we derive the expression for quantum mutual information between the quantum secret extracted by any multi-player structure and the share held by the referee corresponding to the Tyc-Rowe-Sanders continuous-variable quantum secret-sharing scheme. We solve by converting the Tyc-Rowe-Sanders position representation for the state into a Wigner function from which the covariance matrix can be found, then insert the covariance matrix into the standard formula for continuous-variable quantum mutual information to obtain quantum mutual information in terms of squeezing. Our quantum mutual information result quantifies the leakage of  the ramp quantum secret-sharing schemes.	
\end{abstract}
\section{Introduction}
\label{sec:Introduction}
Secret sharing (SS) is an information theoretically secure cryptographic protocol that is applicable to online auctions, electronic voting, shared electronic banking and cooperative activation  in the classical domain~\cite{Iftene2006SecretSS}, and distributed quantum computing  in the quantum regime~\cite{gerd2007quantum}. Ramp classical~\cite{10.1007/3-540-39568-7_20,yamamoto1986secret} and quantum~\cite{PhysRevA.72.032318,PhysRevA.87.032319} secret-sharing (SS) schemes were proposed to reduce the communication complexity by the sacrifice of security conditions. Continuous-variable quantum secret sharing (CV QSS)~\cite{0305-4470-36-27-314,PhysRevA.65.042310,PhysRevLett.92.177903} has been formulated in the framework of discrete-variable quantum SS schemes~\cite{PhysRevLett.83.648}, which does not accommodate the quantum-information leakage inherent in continuous representations of quantum information.
Our aim is to formulate CV QSS  as a continuous-variable (CV) ramp quantum secret sharing (RQSS) protocol and introduce a  technique to certify the protocol. 
	
In order to reach our aims, we introduce four advances in our work. We develop the quantum mutual-information approach to the continuous-variable regime for evaluating the security of CV QSS schemes. We derive quantum mutual information between referee and any multi-player structure corresponding to the Tyc-Rowe-Sanders (TRS03) CV QSS scheme~\cite{0305-4470-36-27-314}. Furthermore, we introduce a certification technique for CV QSS in the framework of quantum-interactive proofs~\cite{kapourniotis2014verified,10.1007/978-3-540-24587-2_1,aolita2015reliable} and demonstrating the necessity of it being a RQSS scheme. Also we give an upper bound for the failure probability in terms of the number of experimental runs from which the referee knows how many rounds are required  to have sufficient information.
	
We focus on the \enquote{quantum-quantum} (QQ) SS schemes~\cite{PhysRevLett.83.648} (in which the secret is a quantum state and communication occurs over quantum channels) because the \enquote{classical quantum} (CQ) SS schemes (which is for sharing a classical message over quantum channels)~\cite{PhysRevA.59.1829,PhysRevA.59.162}, can be simulated by QKD and classical secret sharing~\cite{PhysRevA.78.042309}. The QQ case was extended to  CV regime  by Tyc and Sanders~\cite{PhysRevA.65.042310} and has been    realized experimentally for three players, any two of whom are authorized to extract the secret state~\cite{PhysRevLett.92.177903,PhysRevA.71.033814}.    Importantly, TRS03 later showed that the continuous-variable quantum state sharing could be extended to a $(k, n)$ threshold scheme (a class of QSS schemes  in which the authorized structure consists of  all groups of~$k$ or more players while there are $n$ players in total~\cite{PhysRevLett.83.648}), without a corresponding scale up in quantum resources.
	
Whereas conditional entropy is employed for evaluating the security of CC schemes, quantum mutual information is needed for the quantum case~\cite{Imai:2005:ITM:2011608.2011615}.  Quantum mutual information has been used as a means to evaluate the secrecy condition of Cleve-Gottesman-Lo QSS in the~$(2,3)$~case~\cite{Imai:2005:ITM:2011608.2011615}. TRS03 characterized  the quality of secret extraction for their scheme by calculating the fidelity in terms of squeezing parameter 
between the original and the extracted secret  for an arbitrary coherent state as the secret. However, fidelity is not a distance measure~\cite{wilde2013quantum}. 

Hence, we develop the alternative and more meaningful quantum mutual-information approach for evaluating the  CV QSS security. Restricting to Gaussian states and operations allows all the calculation to be performed within the convenient framework of the semidirect product
\begin{equation}
\label{eq:eqone}
	\text{HW}(n)\rtimes\text{Sp}\left(2n,\mathbb{R}\right),
\end{equation}
which is the continuous-variable Clifford group,
with~$\text{Sp}\left(2n,\mathbb{R}\right)$ the symplectic group
and HW$\left(2n,\mathbb{R}\right)$ the Heisenberg-Weyl group for $n$ modes~\cite{PhysRevLett.88.097904}.
This representation makes calculations tractable but ignores potentially powerful non-Gaussian operations~\cite{Lloyd2003}.

Our paper is organized as follows. In \S\ref{sec:Background}, we briefly review  the theoretical background on continuous-variable quantum information with Gaussian states and Gaussian operations, mutual information and discrete-variable  ramp  quantum  SS protocols. We detail our approach in \S\ref{sec:Approach}. The mathematical results  are presented in \S\ref{sec:Results}. We conclude with a discussion of our results in \S\ref{sec:Discussion}.

\section{Background}\label{sec:Background}
This section provides the  required context to tackle the problem which is solved in this paper. We begin the section by theoretical background on continuous-variable quantum information with Gaussian states and Gaussian operations. Then we discuss quantum mutual information, which is a necessary tool for defining and evaluating quantum SS schemes. Finally, we discuss basic results of RQSS schemes.

\subsection{Continuous-variable quantum information with Gaussian states and Gaussian operations}\label{sec:Continuous-variable quantum information with Gaussian states and Gaussian operations}

In this subsection, we begin by introducing Gaussian states~\cite{RevModPhys.84.621} and some of their important properties. Then we explain the Gaussian preserving maps, which preserve the Gaussian property of quantum states. Finally, we discuss continuous-variable quantum secret sharing based on TRS03 CV QSS scheme.
	
\subsubsection{Gaussian states}
A continuous-variable quantum state is  an continuously parameterized element of  Hilbert space described by observables with continuous eigenspectra.
Typically, a continuous-variable quantum state is described by~$n$ bosonic modes, associated with a tensor-product Hilbert space 	
\begin{equation}
\label{eq:eqtwo}
\mathscr{H}^{\otimes n}
=\bigotimes_{\substack{k=1}}^{k=n}\mathscr{H}_k	\sim \mathcal{L}^2\left(\mathbb{R}^n\right), 
\end{equation}
i.e., square integrable complex-valued functions over $\mathbb{R}^N$
and a vector of quadrature operators
\begin{equation}
\label{eq:eqthree}
	\hat{\bm{x}}\coloneqq (\hat{q}_1,\hat{p}_1,...,\hat{q}_n,\hat{p}_n)^\top
\end{equation}
for ${}^\top$ denoting transpose.
The vector $\hat{\bm{x}}$ satisfies the commutation relation
\begin{equation}\label{eq:eqfour}
\left[\hat{\bm{x}}_i,\hat{\bm{x}}_j\right]	=\Omega_{ij},\ \ \
\bm{\Omega}=\bigoplus_{k=1}^{n}\begin{pmatrix}
0&1\\
-1&0
\end{pmatrix},
\end{equation}
known as the symplectic form.

An arbitrary continuous-variable quantum state is characterized by a density operator 
\begin{equation}\label{eq:eqfive}
\rho \in \mathcal{S}(\mathscr{H}),
\end{equation}
where $\mathcal{S}(\mathscr{H})$ is the set of positive semidefinite trace-class operators.
These positive trace-class operators can be represented by the Wigner function~\cite{PhysRev.40.749} 
\begin{equation}
\label{eq:eqsix}
	W\left(\bm{x}\right)
		=\frac{1}{(2\pi)^{2n}}\int_{\mathbb{R}^{2n}}\text{d}^{2n}\bm{\xi}\exp\left(-\text{i}\bm{x}^\top\bm{\xi}\right)\chi\left(\bm{\xi}\right)
\end{equation}
for
\begin{equation}
\label{eq:eqseven}
\chi(\bm{\xi})\coloneqq\operatorname{tr}\left[ \rho\hat{D}(\bm{\xi})\right],
\end{equation}
being the the Wigner characteristic function and
\begin{equation}
\label{eq:eqeight}
\hat{D}\left(\bm{\xi}\right)\coloneqq\exp\left(\text{i}\bm{x}^\top\bm{\xi}\right),\qquad \bm{\xi}\in\mathbb{R}^{2n}
\end{equation}
being the Weyl operator.
Wigner functions are particularly useful for calculating expectation values of symmetrically ordered functions~$\hat{q}$ and~$\hat{p}$ denoted by~$S\left(\hat{q}^b\hat{p}^d \right)$, with~$S$ denoting symmetric ordering, and with expectation value
\begin{equation}\label{eq:eqnine}
\operatorname{tr}\left[\rho S\left(\hat{q}^b\hat{p}^d \right)\right]=\int \text{d}q \text{d}p\,
W\left(\bm{x}\right) q^b p^d.
\end{equation}
Thus far, we have the Wigner representation for any state;
now we restrict to Gaussian states.
	
A Gaussian state is defined to be a state whose Wigner representation is Gaussian.
A Gaussian state can be completely characterized by its first moment $\bar{\bm{x}}=\operatorname{tr}\big(\hat{\bm{x}}\rho\big)$  and covariance matrix $\bm{V}$.
The covariance matrix entries are
\begin{equation}
\label{eq:eqone0}
V_{ij}\coloneqq \frac{1}{2}\operatorname{tr} \left[\{\bm{\Delta}\hat{\bm{x}}_i,\bm{\Delta}\hat{\bm{x}}_j\}\right],\;
\bm{\Delta}\hat{\bm{x}}_i\coloneqq\hat{\bm{x}}_i-\operatorname{tr}\left(\hat{\bm{x}}_i\rho\right),\; i,j\in\{1,\ldots, 2n\},
\end{equation}
with $\{ , \}$ the anticommutator.

The symplectic manipulation of a Gaussian state's covariance matrix can be used to express its fundamental properties. By definition, a~$2n\times 2n$ real-valued matrix~$\bm{S}$ is called \textit{symplectic} if it preserves the symplectic form of Eq.~(\ref{eq:eqthree}); i.e.,
\begin{equation}\label{eq:eqone1}
\bm{S \Omega S}^\top=\bm{\Omega}.
\end{equation}
According to the Williamson theorem~\cite{simon1999congruences},
each covariance matrix $\bm{V}$ has a corresponding symplectic transformation~$\bm{S}$
satisfying 
\begin{equation}\label{eq:eqone2}
\bm{V}=\bm{S}\left[\bigoplus_{k=1}^{n}\nu_k\bm{I}_k\right]\bm{S}^\top,
\end{equation}\label{eq:13}
with symplectic spectrum defined by the vector
\begin{equation}
\label{eq:eqone3}
\bm{\nu}
\coloneqq\left(\nu_1,\ldots,\nu_n\right)
\end{equation}
unique to each~$\bm{V}$
and satisfying
\begin{equation}
\label{eq:eqtenfour}
\prod_{k=1}^{n}\nu_{k}^{2}
=\det{\bm{V}}.
\end{equation}
As an example,
a two-mode Gaussian state has covariance matrix 
\begin{equation} 
\label{eq:eqtenfive}
\bm{V}=\begin{pmatrix}
\bm{A}&\bm{C}\\
\bm{C}&\bm{B}
\end{pmatrix};\;
\bm{A}=\bm{A}^\top,
\bm{B}=\bm{B}^\top,
\bm{C}\in\mathbb{R}^{2\times2}.
\end{equation}
The symplectic spectrum is~\cite{0953-4075-37-2-L02}
\begin{equation}\label{eq:eqonesix}
\nu_\pm=\sqrt{\frac{\Delta\pm\sqrt{\Delta^2-4\det \bm{V}}}{2}},
\end{equation}
where
\begin{equation}
\label{eq:eqoneseven}
\Delta\coloneqq \det{\bm{A}}+\det{\bm{B}}+2\det{\bm{C}}.
\end{equation}
As Gaussian states are easy to describe mathematically, a large class of transformations acting on such states   are easy to characterize as well.
In the next section, we discuss this class of transformations called Gaussian preserving maps.
	
\subsubsection{Gaussian-preserving maps}\label{sec:Gpm}
Gaussian (linear) unitary Bogoliubov transformations are	interactions that preserve the Gaussian character of a quantum state. In terms of the quadrature operators, a Gaussian map is described by the affine map
\begin{equation}
\label{eq:eqteneight}
\left(\bm{S},\bm{d}\right): \bm{S}\hat{\bm{x}}+\bm{d},\qquad \bm{d} \in \mathbb{R}^{2n},
\end{equation}
where~$\bm{S}$~(\ref{eq:eqone0}) is the matrix representation of the symplectic group. The most general form of a Gaussian map in terms of its action on the statistical moments $\bar{\bm{x}}$ and $\bm{V}$ is
\begin{equation}
\label{eq:eqtennine}
\bar{\bm{x}}\mapsto\bm{S}\bar{\bm{x}}+\bm{d},\;
\bm{V}\mapsto \bm{S}\bm{V}\bm{S}^\top.
\end{equation}
A special class of Gaussian maps are linear canonical point transformations, for which the positions and momenta transform separately and do not mix~\cite{moshinsky1971linear}.

For single-mode squeezing we have the infinite-dimensional unitary representation~\cite{lvovsky2015squeezed} 
\begin{equation}
\label{eq:twenty}
S_1=\text{e}^{\frac{1}{2}\left(\zeta^\star\hat{a}^2-\zeta\hat{a}^{\dagger 2} \right)},
\end{equation}
and for two-mode squeezing we have the infinite-dimensional unitary representation 
\begin{equation}
\label{eq:twone}
S_2=\text{e}^{\frac{1}{2}\left(\zeta^\star\hat{a}_1\hat{a}_2-\zeta\hat{a}_1^{\dagger}\hat{a}_2^{\dagger}\right)},
\end{equation}
where
\begin{equation}
\label{eq:twtwo}
\hat{a}_k=\frac{\hat{q}_k+\text{i}\hat{p}_k}{\sqrt{2}},\;
\hat{a}_k^\dagger=\frac{\hat{q}_k-\text{i}\hat{p}_k}{\sqrt{2}},\;
\zeta=s\text{e}^{\text{i}\theta},\;
s\in\mathbb{R}^+.
\end{equation}
A two-mode squeezed vacuum (TMSV) state is mathematically represented as~\cite{lvovsky2015squeezed}
\begin{equation}
\label{eq:twonezeta}
\ket{\zeta}_\text{TMSV}	
:=S_2\left(\zeta\right)\ket{0},\;
\zeta\in\mathbb{C}.
\end{equation}
In the next section, we explain TRS03 continuous-variable quantum SS scheme in which the Gaussian maps are used for encoding and decoding.

\subsubsection{Continuous-variable quantum secret sharing}
In this subsection, we explain the TRS03 CV QSS scheme. In a $\left(k,2k-1\right)$-threshold scheme, the dealer possesses a pure secret state $\ket{\psi}\in\mathscr{H}$ and encodes the quantum secret into an entangled state of $2k-1$ modes of the electromagnetic field by combining it with $2k-2$ ancillary states. The dealer then distributes them among the $n$ players, each of whom receive one share, and at least~$k$ players must combine their shares in an active interferometer to extract the secret state.
	
Let $\mathscr{H}^{(2k-1)}$ be the tensor product of $2k-1$ copies of $\mathscr{H}^{(1)}$ and each player owns one of these copies.
Let us define $\mathbb{F}^{2k-1}$ as the real linear space of coordinate functions for
$\mathbb{R}^{2k-1}$.
Then a system of Euclidean coordinates
\begin{equation}\label{eq:xinr}
\bm{x}=\left(x_1, x_2,\ldots,x_{2k-1}\right)^\top\in \mathbb{R}^{2k-1},
\end{equation}
is equivalent to choosing an orthonormal basis
of coordinate functions
\begin{equation}
	\bm{f}
		:=\left(f_1,f_2,\ldots,f_{2k-1}\right)^\top\in\mathbb{F}^{2k-1}
\end{equation}
such that 
\begin{align}
	f_i:\mathbb{R}^{2k-1}\to\mathbb{R}:
		\left(\bm{x}\right)=x_i
\end{align}
with~$x_i$ the $i^\text{th}$ coordinate of~$\bm{x}$~(\ref{eq:xinr}),
and $f_i\cdot f_j=\delta_{ij}$.
	
Initially, the dealer starts with an unentangled tensor product 
\begin{equation}
\ket{\Psi}
=\ket{\psi}\otimes \underbrace{\ket{\phi_a}\otimes\cdots\otimes\ket{\phi_{a}}}_{k-1}
\otimes\underbrace{\ket{\phi_{1/a}}\otimes\cdots\otimes\ket{\phi_{1/a}}}_{k-1},
\end{equation}
where $\ket{\psi}$ is the secret state and 
\begin{equation}
\phi_a\left(x\right)=\braket{x|\phi_a}
=\left(\pi a^2\right)^{-1/4} \text{e}^{-x^2/{2a^2}}.
\end{equation}
Let us write this state as 
\begin{align}
\ket{\Psi}
=&\int\text{d}^n\bm{x}\, \Psi(\bm{x})\ket{x_1}\otimes \cdots\otimes \ket{x_n}
\nonumber\\
=&
\int \text{d}^n\bm{x}\ \Psi  \left(\bm{x}\right)\ket{f_1(\bm{x})}\otimes \cdots\otimes \ket{f_n(\bm{x})},
\end{align}
where
\begin{equation}
\ket{\Psi}=\psi\left(x_1\right)\prod_{i=2}^{k}\phi_a\left(x_i\right)\prod_{i=k+1}^{n}\phi_{1/a}\left(x_i\right).
\end{equation}
	
The dealer then performs the encoding using a linear canonical point transformation
\begin{equation}
f_j \mapsto g_i= \sum_j g_{ij} f_j.
\end{equation}
The corresponding unitary transformation  then maps the state $\ket{\Psi}$ to
\begin{equation}
\left|\det g \right|^{1/2}
	\int \text{d}^{2k-1}\bm{x} \Psi  \left(\bm{x}\right)\ket{g_1(\bm{x})}\otimes \cdots\otimes \ket{g_{2k-1}(\bm{x})}.
\end{equation}
The dealer, however, has to choose $\{g_i\}$ such that any~$k$ players are able to disentangle the secret state but that fewer is unable to do so. For this purpose, in the case of sufficiently large~$a$,
only the orthogonal projection  $\iota_i$ of each vector $g_i$ into the space spanned by the vectors $\{f_1, \ldots,f_{2k-1}\}$ is important. The vectors $\{g_i\}$ then must be chosen such that any~$k$ vectors from the set $\{f_1, \iota_1, \ldots, \iota_{2k-1} \}$ are linearly independent. This linear independence condition guarantees that any~$k$ players are able to extract the secret.
	
For convenience, let us express~$\mathbb{F}^{2k-1}\in \mathbb{R}^{2k-1}$ as a direct sum of three mutually orthogonal subspaces
\begin{equation}
\mathbb{F}^{2k-1}=\mathbb{X}\oplus \mathbb{Y}\oplus \mathbb{Z},
\end{equation}
where~$\mathbb{X}$ is the one-dimensional space spanned by~$f_1$ and~$\mathbb{Y}$ and~$\mathbb{Z}$ are~$k-1$-dimensional spaces spanned by~$\{f_2, \ldots, f_k\}$ and~$\{f_{k+1},\ldots, f_{2k-1} \}$, respectively. Now let us relabel~$\{x_i\}$ coordinates as~$(x,y_i,z_i)$ coordinates with 
\begin{align}
x=&x_1,\; y_i=x_{i+1},\; z_i=x_{k+i},
i\in\{1,\ldots,k-1\}.
\end{align}
The wavefunction $\Psi$ is then 
\begin{equation}
\Psi\left(\bm{x}\right)=\psi\left(x\right)\prod_{i=1}^{k-1} \phi_a\left(y_i\right)\phi_{1/a}(z_i).
\end{equation}

Without loss of generality, the first~$k$ players collaborate to retrieve the quantum secret. The players then make the linear coordinate transformation
\begin{equation}
g_i\mapsto\xi_i=\sum_j \xi_{ij} f_j
\end{equation}
assuming $\xi_i=g_i$ for all $i>k$.
	
For convenience, let us define a decomposition for every vector $\xi_i$ as a sum of three mutually orthogonal vectors, each of which belongs to subspaces $\mathbb{X}$, $\mathbb{Y}$ and $\mathbb{Z}$
\begin{equation}
\label{eq:encdectrs03}
\xi_i=\alpha_i+\beta_i+\gamma_i.
\end{equation} 
Equivalently,  we can write
\begin{equation}
\label{eq:encdectrs032}
\xi_i\left(\bm{x}\right)=\alpha_i x+\sum_j \beta_{ij} y_j+\sum_j \gamma_{ij}z_j.
\end{equation}
In the case that the vectors $g_i$ are chosen in such a way that any~$k$ vectors from the set $\{f_1, \iota_1, \ldots, \iota_{2k-1}  \}$
are linearly independent, the players can design the transformation
$g_i\mapsto\xi_i$
such that
\begin{align}\label{eq:dectrans}
\alpha_1=&1, \qquad \beta_1=0,\nonumber \\ \alpha_{i+1}=&\alpha_{k+i}, \; \beta_{i+1}=\beta_{k+i},
\end{align}
where $i\in\{1,\ldots,k-1\}$.
Then transformation~(\ref{eq:dectrans})
extracts the secret for sufficiently large values of parameter $a$.

\subsection{Mutual information}\label{sec:Mutual information}
Here we review the key notions of mutual information,
which is the method for quantifying information security and defining quantum secret sharing.
We begin by presenting salient facts about Shannon and von~Neumann entropy followed by requisite knowledge concerning classical and quantum mutual information.
Finally, in this subsection,
we discuss the security for discrete quantum secret sharing as our aim is to analyze security
for continuous-variable quantum secret sharing.
\subsubsection{Shannon and von~Neumann entropy}
Here we review Shannon and von~Neumann entropy as these notions of entropy underpin the formulation of classical and quantum mutual information.
This subsubsection also helps to elucidate the compact notation we use throughout this paper.
\paragraph{Shannon entropy.}
Let~$Z$ be a statistical ensemble defined by a classical random variable~$z$
and its associated probability distribution $\{p_j\}= \{p_1,\ldots,p_n\}$,
which can be expressed as a probability vector
$\bm{p}=(p_1,\ldots,p_n)^\top$.
The logarithm of this vector (always using base~2 here) is
\begin{equation}
     \log\bm{p}
            :=(\log\bm{p}_j).
\end{equation}
Using the Hadamard (elementwise) product~$\bm{a}\circ\bm{b}:=(a_ib_i)$~\cite{Davis1962} for vectors
and the sum of such elements $\bm{a}\odot\bm{b}:=\sum_ia_ib_i$,
the Shannon entropy is
\begin{equation}
\label{eq:HSh}
	H_\text{Sh}(\bm{p})
		=-\bm{p}\odot\log\bm{p}
		=-\bm{p}\cdot\bm{\log p}.
\end{equation}
Thus,
$H_\text{Sh}$ yields the number of bits per letter needed to completely specify~$Z$ in the asymptotic limit of infinitely long strings~\cite{Shannon:2001:MTC:584091.584093}.
Shannon entropy is thus a measure for the uncertainty of~$z$
or it indicates how much information each letter in the string that uses the alphabet~$Z$ carries.
\paragraph{Von Neumann entropy}
In the same vein, the information content of a quantum state~$\rho$~(\ref{eq:eqfive}) can be quantified  by determining how many qubits are needed to represent state~$\rho$ in the asymptotic limit of an infinite ensemble of physical systems.
This quantum-information content,
known as the von~Neumann entropy~\cite{von1955mathematical},
amounts to computing a classical Shannon entropy~(\ref{eq:HSh})
\begin{equation}
\label{eq:HvN}
H_\text{vN}(\rho)
=-\operatorname{tr}\left(\rho\log_2\rho\right)
=H_\text{Sh}\left(\operatorname{spec}\bm{\rho}\right),
\end{equation}
for $\operatorname{spec}\bm{\rho}$
a vector comprising eigenvalues of the state~$\rho$.
\paragraph{Continuous-variable quantum entropy.}
For continuous-variable Gaussian states,
we define the vectors
\begin{equation}
\label{eq:eqfour3}
\bm{\nu}^\pm
:=\frac{\bm{\nu}\pm\mathds{1}}{2}
\end{equation}
with~$\bm{\nu}$ the symplectic spectrum~(\ref{eq:eqone3})
and~$\mathds{1}$ the vector with all entries being unity.
Thus,
the von~Neumann entropy is~\cite{PhysRevA.59.1820}
\begin{equation}
\label{eq:entropy}
	H_\text{vN}(\rho)
		=\bm{\nu}^+\odot\log\bm{\nu}^+
			+\bm{\nu}^-\odot\log\bm{\nu}^-.
\end{equation}
These entropy expressions are used in the formul\ae\ for mutual information.
\paragraph{Convenient notation for states in entropy formul\ae}
A convenient notation for entropy,
which is independent of being classical or quantum,
uses a label for the classical or quantum state.
Rather than specify the state as~$\bm{p}$ classically or~$\rho$ quantumly,
we label the state by a capital letter such as~A and~B,
with these labels commensurate with the usual Alice-and-Bob nomenclature in cryptology~\cite{RSA78}.
\paragraph{Conditional entropy.}
Labelling the joint state held by~A and~B as~$\text{AB}$, the conditional entropy is abstractly expressed as
\begin{equation}
\label{eq:conditionalentropy}
H\left(\text{A}|\text{B}\right)\coloneqq H\left(\text{A}\text{B}\right)-H\left(\text{B}\right)
\end{equation}
for any valid formula for entropy, whether classical~(\ref{eq:HSh}) or quantum~(\ref{eq:HvN}).
\paragraph{Classical conditional entropy.}
\label{para:classicalconditionalentropy}
The classical conditional entropy~\cite{cover2012elements} is obtained from Eq.~(\ref{eq:conditionalentropy})
by replacing
\begin{equation}
	H(\text{A})\mapsto
	H_\text{Sh}\left(\bm{p}_\text{A}\right)
\end{equation}
for~$\bm{p}_\text{A}$
the distribution held by~A.
Similarly, we replace
\begin{equation}
	H(\text{B})\mapsto
	H_\text{Sh}\left(\bm{p}_\text{B}\right)
\end{equation}
and
\begin{equation}
	H(\text{AB})\mapsto
	H_\text{Sh}\left(\bm{p}_\text{AB}\right).
\end{equation}
$H\left(\text{A}|\text{B}\right)$ quantifies the correlation between~$\text{A}$ and~$\text{B}$ as the reduction of the number of bits per letter needed to specify~$\text{A}$ given~$\text{B}$ is known.
\paragraph{Quantum conditional entropy.}
\label{para:quantumconditionalentropy}
The quantum conditional entropy~\cite{wilde2013quantum} is obtained from Eq.~(\ref{eq:conditionalentropy})
by replacing
\begin{equation}
	H(\text{A})\mapsto
	H_\text{vN}\left(\rho_\text{A}\right)
\end{equation}
for~$\rho_\text{A}$
the quantum state held by~A.
Similarly, we replace
\begin{equation}
	H(\text{B})\mapsto
	H_\text{vN}\left(\rho_\text{B}\right)
\end{equation}
and
\begin{equation}
	H(\text{AB})\mapsto
	H_\text{vN}\left(\rho_\text{AB}\right).
\end{equation}
Although classical conditional entropy is always positive,
for evaluatingquantum conditional entropy can be negative~\cite{horodecki2005partial}.  
\subsubsection{Classical and quantum mutual information}
We explain classical mutual information~\cite{cover2012elements} and
quantum mutual information~\cite{wilde2013quantum},
first as an abstract concept regardless of whether classical or quantum information is chosen.
Then we explain each of classical and quantum mutual information.
Quantum mutual information is vital for evaluating security for secret sharing.
\paragraph{Mutual information.}
Labelling the joint state held by~A and~B as~$\text{AB}$,
mutual information is abstractly expressed as
\begin{equation}
\label{eq:mutualinformation}
I\left(\text{A};\text{B}\right)
\coloneqq H\left(\text{A}\right)+H\left(\text{B}\right)-H\left(\text{AB}\right)
\end{equation}
for any valid formula for entropy,
whether classical~(\ref{eq:HSh}) or quantum~(\ref{eq:HvN}).
Classical mutual information~\cite{wilde2013quantum} is obtained from Eq.~(\ref{eq:mutualinformation})
by replacing
\begin{equation}
	H(\text{X})\mapsto
	H_\text{Sh}\left(\bm{p}_\text{X}\right) 
\end{equation}
with~$\text{X}\in\{\text{A},\text{B}\}$ for~$\bm{p}_\text{X}$ and
\begin{equation}
	H(\text{AB})\mapsto
	H_\text{Sh}\left(\bm{p}_\text{AB}\right)
\end{equation}
as discussed in~\P\ref{para:classicalconditionalentropy}.
Classical mutual information quantifies the correlation between two statistical ensembles~$\text{A}$ and~$\text{B}$
as the reduction of the number of bits per letter needed to specify one of the variables given the other variable is known.
\paragraph{Quantum mutual information.}
The quantum mutual information~\cite{wilde2013quantum} is obtained from Eq.~(\ref{eq:mutualinformation})
by replacing
\begin{equation}
	H(\text{A})\mapsto
	H_\text{vN}\left(\rho_\text{A}\right)
\end{equation}
for~$\rho_\text{A}$
the quantum state held by~A.
Similarly, we replace
\begin{equation}
	H(\text{B})\mapsto
	H_\text{vN}\left(\rho_\text{B}\right)
\end{equation}
and
\begin{equation}
	H(\text{AB})\mapsto
	H_\text{vN}\left(\rho_\text{AB}\right).
\end{equation}
Quantum mutual information is always positive and quantifies the total correlations contained in the bipartite state~$\rho_\text{AB}$. 
Quantum mutual information is employed to define and evaluate the security of quantum secret-sharing schemes (QSS).
\paragraph{Relation between conditional entropy and mutual information.}
The relation between conditional entropy and mutual information is
\begin{equation}
\label{eq:RMIH}
I\left(\text{A};\text{B}\right)= H\left(\text{A}\right)-H\left(\text{A}|\text{B}\right)=H\left(\text{B}\right)-H\left(\text{B}|\text{A}\right)
\end{equation}
for any valid formula for entropy,
whether classical~(\ref{eq:HSh}) or quantum~(\ref{eq:HvN}).
The relation between classical mutual information
and classical conditional entropy is obtained from Eq.~(\ref{eq:RMIH})
by replacing
\begin{equation}
	H(\text{X}) \mapsto
	H_\text{Sh}\left(\bm{p}_\text{X}\right) 
\end{equation}
with~$\text{X}\in \{\text{A},\text{B}\}$ and
\begin{equation}
	H(\text{X}|\text{Y}) \mapsto
	H_\text{Sh}\left(\bm{p}_\text{XY}\right)-H_\text{Sh}\left(\bm{p}_\text{Y}\right) 
\end{equation}
with~$\left(\text{X},\text{Y}\right)\in \{\left(\text{A},\text{B}\right),\left(\text{B},\text{A}\right)\}$ as discussed in~$\P$\ref{para:classicalconditionalentropy}.

The relation between quantum mutual information and quantum conditional entropy is obtained from Eq.~(\ref{eq:RMIH})
by replacing
\begin{equation}
	H(\text{X})\mapsto
	H_\text{vN}\left(\rho_\text{X}\right)
\end{equation}
with~$\text{X}\in \{\text{A},\text{B}\}$ and
\begin{equation}
	H(\text{X}|\text{Y})\mapsto
	H_\text{vN}\left(\rho_\text{XY}\right)-H_\text{vN}\left(\rho_\text{Y}\right)
\end{equation}
with~$\left(\text{X},\text{Y}\right)\in \{\left(\text{A},\text{B}\right),\left(\text{B},\text{A}\right)\}$ as discussed in~$\P$\ref{para:quantumconditionalentropy}.
\subsubsection{Classical and quantum secret sharing}
In this subsubsection,
we explain classical and quantum secret-sharing protocols. We begin by establishing the agents of the protocol namely dealer and players and the structures corresponding to the set of players. Afterwards, we explain classical secret-sharing schemes along with classical secrecy and recoverability conditions corresponding to them. Then we define quantum secret sharing and provide the secrecy and recoverability conditions corresponding to them based on quantum mutual information.
\paragraph{Dealer and players.}
\label{para:dealerandplayers}
We establish the agents of the protocol and the structures corresponding to sets of players,
who are one kind of agent.
Specifically,
secret sharing comprises $n+1$ agents,
namely one dealer~$\mathcal{D}$ and~$n$ players labelled
\begin{equation}
\label{eq:players}
\mathcal{P}
=\{P_1, P_2,\ldots,P_n\}.
\end{equation}
The power set of players is~$2^\mathcal{P}$,
which is the set of all subsets of the set of players~(\ref{eq:players}).

The role of the dealer is to encode the secret message
$S\in\{0,1\}^*$
(classically) or $\rho_s\in\mathcal{S}\left(\mathscr{H}\right)$ (\ref{eq:eqfive})
quantumly, into $n$ shares and distributes them among players in such a way that specific elements of $2^\mathcal{P}$ form the authorized structure~$\mathcal{A}$
to retrieve the secret message whereas other elements are denied any information about the secret whatsoever.
The set of elements that are denied any information is known as the forbidden structure~$\mathcal{F}$.
\paragraph{Access structure.}
Let 
\begin{equation}
	\mathcal{F},\,\mathcal{A}\subseteq 2^\mathcal{P},\;
	\mathcal{F},\mathcal{A}\neq\emptyset,
\end{equation}
where~$\mathcal{F}$ is monotonically decreasing
and~$\mathcal{A}$ is monotonically increasing, and
\begin{equation}
\mathcal{F}\cap\mathcal{A}=\emptyset.
\end{equation}
Then the set
\begin{equation}
\label{eq:accessstructure1}
	\Gamma=\{\mathcal{F},\mathcal{A}\}
\end{equation}
is the access structure on $\mathcal{P}$. 
Quantumly,
the no-cloning theorem implies that the existence of two disjoint authorized group is forbidden~\cite{PhysRevA.61.042311}.

\paragraph{Secret-sharing protocol.}
 Let $\mathscr{H}$ be a Hilbert space and let $\mathcal{S}(\mathscr{H})$ be all density operators on a Hilbert space $\mathscr{H}$.
 In a quantum secret-sharing scheme, the dealer's task is to encrypt a quantum secret $\rho_\text{s}\in \mathcal{S}\left(\mathscr{H}\right)$ into a composite system of Hilbert spaces
\begin{equation}
\mathscr{H}_1, \mathscr{H}_2, \ldots, \mathscr{H}_n,
\end{equation}
each of which is called a share labelled by $S_1,S_2,\ldots,S_n$. Let 
\begin{equation}
N\coloneqq \{S_1,S_2,\ldots,S_n \}
\end{equation}
be the entire set of shares and
\begin{equation}
\mathscr{H}_N:=\bigotimes_{S_i\in N}\mathscr{H}_{S_i}
\end{equation}
be the corresponding Hilbert space. The dealer then distributes the shares among players~(\ref{eq:players}).  For a subset $A \subseteq N$ of shares
\begin{align}
\mathscr{H}_\text{A}\coloneqq\bigotimes_{S_i\in A}\mathscr{H}_{S_i},
\end{align}
the QSS encoding is
\begin{align}\label{QSSQO}
W_N: \mathcal{S}(\mathscr{H}) \to \mathcal{S}(\mathscr{H}_N),
\end{align}
which is a completely positive and trace preserving map~\cite{PhysRevA.72.032318}.

The composition map of the encoder $W_N$ for a subset $X \subseteq N$,
and the partial trace of the complement $N\setminus X$ is
\begin{align}
W_X\coloneqq \operatorname{tr}_{N\setminus X}W_N.
\end{align}
A QSS scheme is then defined by the quantum operation $W_N$ (\ref{QSSQO}) that is reversible with respect to $\mathcal{S}(\mathscr{H})$.
The set $N$ is divided into two mutually disjoint structures
$\mathcal{A}$ and $\mathcal{F}$~\cite{PhysRevA.72.032318}.
\begin{enumerate}[(i)]
\item A set $X \subseteq N$ is \textit{authorized} if $W_X$ is reversible  with respect to $\mathcal{S}(\mathscr{H})$.
\item  A set $X \subseteq N$ is \textit{forbidden} if $W_X$ is vanishing  with respect to $\mathcal{S}(\mathscr{H})$.
\end{enumerate}
The arguments so far are valid in the classical cases,
which is verified by replacing the corresponding notions with the classical ones~\cite{PhysRevA.72.032318}. 
\paragraph{Classical secrecy and recoverability conditions.} Classical secrecy is expressed in terms of conditional entropy  but equivalently can be expressed in terms of mutual information. Strictly speaking, conditional entropy is between shares. However, for simplicity, in the literature there is a tendency to refer to conditional entropy between players. $\Pi$ is a perfect SS scheme on $\Gamma$ if 
\begin{itemize}
\item $\forall \mathcal{B}\in\mathcal{A}\,\, H\left(S|\mathcal{B}\right)=0$
\item $\forall \mathcal{B}\notin\mathcal{A}\,H\left(S|\mathcal{B}\right)=H\left(S\right)$.
\end{itemize}
\paragraph{Quantum secrecy and recoverabiliy conditions.}
Here we discuss quantum secrecy conditions in terms of quantum mutual information. Strictly speaking, quantum mutual information is between shares. However, for simplicity, in the literature there is a tendency to refer to quantum mutual information between players. We can imagine that the system~$\rho_\text{s}$  is part of a larger system and that this compound system is initially in a pure state $\ket{\psi^\text{RS}}$.
Therefore,
\begin{equation}\label{rhostr}
\rho_\text{s}
=\operatorname{tr}_\text{R}\left(\ket{\psi^\text{RS}}
\bra{\psi^\text{RS}}\right).
\end{equation}

In a QSS, if a subset $X\in 2^\mathcal{P}$ satisfies
\begin{equation}
\label{secrecy condition}
	I\left(\text{R};\text{X}\right)=0 \; (\text{secrecy condition}),
\end{equation}
then $\rho^X$ does not contain any information about~$\rho_s$~\cite{Imai:2005:ITM:2011608.2011615}. On the other hand, if a subset $X$ satisfies	 
\begin{equation}\label{eq:recoverability condition}
I\left(\text{R};\text{X}\right)=I\left(\text{R};\text{S}\right)\; (\text{recoverability condition}),
\end{equation}
then $X$ contains full information about~$\rho_\text{s}$~\cite{Imai:2005:ITM:2011608.2011615}.
\paragraph{Access structure.} Specific subsets of players form the authorized structure
\begin{equation}
\label{eq:accessstructure}
\mathcal{A}
:=\left\{Y\in2^\mathcal{P};
I(\text{R};\text{S})=I(\text{R};\text{X})\right\}
\end{equation}
to retrieve the  message whereas the other subsets,
i.e., the forbidden structure
\begin{equation}
\label{eq:forbiddenstructure}
\mathcal{F}
:=\left\{X\in2^\mathcal{P};
I(R;X)=0\right\},
\end{equation}
are denied any information about the secret whatsoever.
We define the QSS access structure as
\begin{equation}
\label{eq:QSSaccessstructure}
\Gamma:=\left\{\mathcal{A},\mathcal{F}\right\}.
\end{equation}

\paragraph{Threshold secret sharing.}
\label{para:thresholdsecretsharing}
 $((k,n))$ threshold QSS schemes are a class of QSS schemes in which the authorized structure comprises all groups of~$k$ or more players while there are $n$ players in total (the use of double parentheses distinguishes it from a classical scheme). $((k,n))$ quantum threshold schemes exists provided no-cloning theorem is satisfied~\cite{PhysRevA.61.042311}. Any quantum secret sharing scheme can be reduced to $((k,2k-1))$ threshold schemes~\cite{PhysRevA.61.042311}.  In QSS schemes, the size of shares allocated to each player must be at least as large as the size of the secret~\cite{PhysRevA.61.042311,PhysRevA.72.032318}.
\subsection{Ramp quantum secret-sharing scheme}
\label{subsec:RQSSS}
As an extension of $(k,n)$--threshold SS schemes
discussed in $\P$\ref{para:thresholdsecretsharing},
ramp secret-sharing (RSS) schemes were proposed by Blakley-Meadows~\cite{10.1007/3-540-39568-7_20} and Yamamoto~\cite{yamamoto1986secret}.
In RSS  schemes, the dimension of each share is reduced
compared to that of the original system by sacrifice security for admitting the intermediate property for some sets of shares, which are denoted as intermediate sets.

In a $(k,L,n)$ threshold RSS scheme, any~$k$ or more players are able to fully reconstruct the secret~$s$,
whereas any $k-L$ or less players are denied to obtain any information of it. Furthermore, from arbitrary $k-j$ shares for $j\in\{1,\ldots,L-1\}$,
some information of the secret leak out with the size of~$\frac{j}{L}$ in~$s$.
	
A QSS scheme~$W_N$ is called perfect if any set $X \subseteq N$ is either authorized or forbidden. Otherwise, $W_N$ is a RQSS scheme.
The access structure of a RQSS scheme is the list of the forbidden, intermediate, and authorized sets. A set $X \subseteq N$ is called \textit{intermediate} if $W_X$ is neither vanishing nor reversible  with respect to $\mathcal{D}(\mathscr{H})$~\cite{PhysRevA.72.032318}. 
Formally, the access structure of the set $N$ is defined by a map
\begin{equation}
\label{eq:eq64}
f : 2^\mathcal{P} \to \{0,1,2\},
\end{equation}
where~$0, 1$ and~$2$ represent
$\mathcal{F}$,$\mathcal{I}$ and~$\mathcal{A}$, respectively. Now that we have the essential background, we proceed in the next section to explain our approach to CVRQSS.
\section{Approach}
\label{sec:Approach}
In this section,
we introduce a CV RQSS protocol and explain how to certify. We discuss the success criterion of the certification protocol. Furthermore, we specify what the parties need to do to complete the certification.
	
\subsection{Continuous-variable ramp quantum secret-sharing protocol with Gaussian states and operations}\label{ARQSSP}
Here we modify the discrete-variable RQSS protocol discussed in \S\ref{subsec:RQSSS} 
into a continuous-variable counterpart.
We choose Gaussian states and operations,
which are convenient mathematically due to the elegance of techniques based on the semidirect product of the symplectic group and the Heisenberg-Weyl group~(\ref{eq:eqone}).
However, the price paid for this convenience is discarding potentially powerful universal operations~\cite{Lloyd2003}.
Whereas, in the discrete case, specification of number of players and threshold condition $L$ suffices to determine the cardinality of the three structures,
the CV case is more complicated due to squeezing limitations. 
\subsubsection{Quantum-optical resources}	
The optical realization comprises displacers that generate Heisenberg-Weyl group elements and single-mode squeezers, passive beam-splitters and phase-shifters  that generate the  semidirect product of the symplectic group~~(\ref{eq:eqone}).
The inputs are vacuum states of light.
For the closed disk
\begin{equation}
\label{eq:disk}
	D_s
		:=\left\{\zeta\in\mathbb{C}:|\zeta|\leq s^2\right\},\;
	s\in\mathbb{R}^+,
\end{equation}
the dealer's and players' single-mode squeezers~(\ref{eq:twenty})
have limited squeezing capability corresponding to~$\zeta\in D_s$,
with~$s=s_\text{max}^\text{D}$ for the dealer and~$s=s_\text{max}^\text{P}$ for the player.
\subsubsection{Dealer's task}
Here we specify the dealer's task in the RQSS protocol. Dealer's tasks include preparing a quantum secret,
choosing an access structure,
encoding the quantum secret and
distributing shares. 
\paragraph{Two-mode squeezed-vacuum source.}
The dealer prepares a TMSV state~(\ref{eq:twonezeta})
drawn randomly from the uncountable set
\begin{equation}
\label{eq:QD}
Q_\text{D}
:=\left\{\ket{\zeta}_\text{TMSV};\zeta\in D_{s_\text{max}^\text{D}}\right\}.
\end{equation}
The dealer's task is to encode one mode of this quantum state into an $n$-mode entangled state by mixing it with $n-1$ ancillary states in an $n$-mode active interferometer. The dealer then sends one share to each of the players  in such a way that the elements of power set of players are divided into   three predetermined mutually disjoint sets known as authorized, intermediate and forbidden structures.

In order for the dealer to prepare  the TMSV randomly, first, he needs to decide the complex two-mode squeezing parameter $\zeta=s\text{e}^{\text{i}\theta}$~(\ref{eq:twtwo}),
where $s$ is bounded by $s_\text{max}^\text{D}$. The dealer generates two random numbers~$a,b\in\left[ 0,1 \right]$.
Then the dealer assigns
\begin{equation}
	s\gets\sqrt{2a s_\text{max}^\text{D}},\;
	\theta\gets2\pi b.
\end{equation}
\paragraph{Choosing a useful, feasible access structure.}
\label{para:choosingF}
The dealer chooses an access structure~$\Gamma$ based on the desired application. The dealer then runs an algorithm that accepts~$\Gamma$, covariance matrix of TMSV state $\bm{V}$, $s_\text{ max}^\text{D}$ and $s_\text{ max}^\text{P}$  as input and yields the encoding transformation or else null as the output.
The dealer then performs the encoding transformation and distributes the shares among players.  
\subsubsection{Players' task}
\label{subsubsec:playerstask}
The players' task in any authorized set is to reconstruct the quantum secret. One player is assigned to hold the secret after reconstruction. The aforementioned player forms a structure with other players in the authorized set who perform a Gaussian unitary operation on their shares such that the state of the share belonging to the assigned player become the same as the original secret state. The players in any intermediate set are allowed to partially reconstruct the secret state. Furthermore, the players in a forbidden structure should not gain any information about the quantum secret whatsoever.
\subsection{Certification protocol}
\label{subsec:certificationprotocol}
In this subsection we introduce a certification protocol that ascertains whether the RQSS protocol succeeds. The success criterion is discussed in this subsection. We specify what the parties need to do to complete the certification.
\subsubsection{Agents and resources}
In this subsection, we establish
the agents of the certification protocol, namely, the dealer, the players and the referee who serves as skeptical certifier. Furthermore, we specify available resources for each party.

The dealer and players share trusted error-free classical and quantum communication channels between each other,
and the referee also shares trusted error-free classical and quantum communication channels with each player and with the dealer.
In our continuous-variable setting,
the referee possesses single-mode homodyne detectors~\cite{RevModPhys.84.621}.
Henceforth, we only refer explicitly to homodyne measurement, without loss of generality.
The dealer possesses a classical computer to choose the access structure~$\Gamma$
discussed in $\P$\ref{para:choosingF},
and the referee possesses a classical computer to run the certification algorithm.
\subsubsection{Dealer's encoding and announcement}
The dealer chooses
an access structure~$\Gamma$ discussed in~$\P$\ref{para:choosingF}
and announces~$\Gamma$
to the players and to the referee.
The dealer encodes shares based on the choice of~$\Gamma$ and the quantum secret,
such as a randomly chosen state in the parameter disk~(\ref{eq:QD}),
and announces this encoding to the players.
\subsubsection{Rounds}
\label{subsubsec:rounds}
In this subsubsection,
we define `rounds',
which are repetitions of the protocol between the dealer, players and referees.
The concept for these rounds is depicted in Fig.~\ref{fig:strategy}.
\begin{figure}[h]
\begin{center}
\includegraphics[scale=0.25]{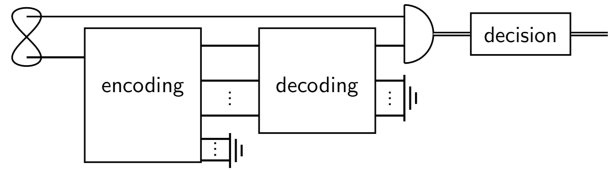}
\end{center}
\caption{Two-mode entangled state with one share, or mode,
sent directly to the referee 
and the other share encoded for the players.
The referee requests a subset of players to decode their shares
and send this result to the referee
who decides whether they have succeeded or not.}
\label{fig:strategy}
\end{figure}
First the dealer prepare a suitable two-mode Gaussian state,
which is the same two-mode Gaussian state for all rounds,
and sends one mode to the referee
and the other mode into an encoder,
which is also unchanging over all rounds.
This encoder creates shares 
that are sent to each player.

After the shares are received by players,
the referee requests a subset of players,
which can be authorized, forbidden or intermediate,
to try to reconstruct the quantum secret and then send their shares to the referee.
The referee then performs single-mode homodyne measurements and save the measurement results.
Rounds continue until the referee permits the dealer and players to stop.
\subsubsection{Referee's certification strategy}
\label{subsubsec:certificationstrategy}
The referee's task is to certify the protocol by ascertaining the dealer's announcement that the access structure is the announced~$\Gamma$.
The referee conducts tests by requiring many rounds per instance,
with each instance corresponding to testing
whether a fixed subset of players is in~$\mathcal{A}$,~$\mathcal{I}$ or~$\mathcal{F}$ structures determined by~$\Gamma$.
Due to the statistical nature of the test,
the referee cannot be 100\% sure that the inference is correct;
rather the referee makes a decision if the probability of being correct exceeds some threshold value,
itself strictly greater than~$1/2$.
\paragraph{Sufficiency condition.}
When a sufficiency condition is met to ascertain whether the subset of players are determined to be in a structure compatible with the dealer's announced~$\Gamma$,
the referee instructs the players to stop.
If that instance passes the test,
the referee announces a new subset of players to test and the rounds repeat until the referee has enough data to pass the sufficiency test.
If the instance results in the dealer and players failing,
the procedure stops as the team of dealer and players has failed the test.
The dealer and players pass only if every instance passes.

\subsection{Summary of approach }
Here we modified the discrete-variable RQSS protocol as the CV counterpart in the case of Gaussian states and operations. Furthermore,  we introduced a certification protocol that ascertains whether the RQSS protocol succeeds. Also we discussed the success criterion and we specified what the parties need to do to complete the certification.
	
\section{Results}
\label{sec:Results}
In this section we present our main results.
Our first result is a CV version of quantum mutual information.
This CV quantum mutual information is then used to quantify quantum-information leakage for Gaussian states and operations.
Based on this leakage characterization,
we introduce a certification test,
in the framework of quantum-interactive proofs,
and provide a practical test to implement this test.
\subsection{CV quantum mutual information}
\label{subsec:CVQMI}
In this subsection, we develop the quantum mutual information for the CV RQSS quantum access structures and employ it to quantify quantum-information leakage for Gaussian states and operations.
We define~$\mathcal{I}$
corresponding to CV RQSS protocols based on quantum mutual information.

Let $\ket{\psi}^\text{RS}$  be a pure two-mode Gaussian state and let the  quantum secret   be~$\rho_\text{s}$~(\ref{rhostr}).
Then
\begin{align}
\mathcal{I}\coloneqq& \left\{X\ ;\; 0<I\left(\text{R};\text{X}\right)<I\left(\text{R};\text{S}\right)\right\},
\end{align}
and ~$\mathcal{A}$ and~$\mathcal{F}$ are  obtained from Eqs.~(\ref{eq:accessstructure}) and~(\ref{eq:forbiddenstructure}), respectively.
	
We now calculate mutual information between the referee and any multiplayer structure for TRS03. Specifically, we consider a two-mode entangled state~(\ref{eq:disk}) such that one mode is used for the secret and the other mode is used for the reference system. We choose this system because that way the referee can do a sensitive entanglement check to verify that the reconstructed state is entangled with a reference system as it should be.
To simplify matters, without loss of generality,
we investigate in particular a TMSV with one mode being the quantum secret and the other mode being the reference system.

We solve the quantum mutual information between an extracted secret 
obtained by any player structure with~$k$ elements and the reference system. In order to do so, by using Eq.~(\ref{eq:eqfive}),
we transform the density function of the reference system and the extracted secret~(\ref{eq:position representation}) into a Gaussian Wigner function represented by a mean vector and a covariance matrix
from which the symplectic eigenvalues~(\ref{eq:eqone3}) are calculated. 
	
The symplectic eigenvalues~(\ref{eq:eqone3}) are inserted into Eq.~(\ref{eq:entropy})
in order to calculate the local and global von Neumann entropy of the extracted secret and reference system from which the quantum mutual information is solved (\ref{eq:mutualinformation}).
Figure~\ref{fig:RQSS} shows the resultant quantum mutual information versus squeezing parameter in the case of $|\zeta|=2$.
\begin{figure}[h]
\begin{center}
\includegraphics[scale=0.5]{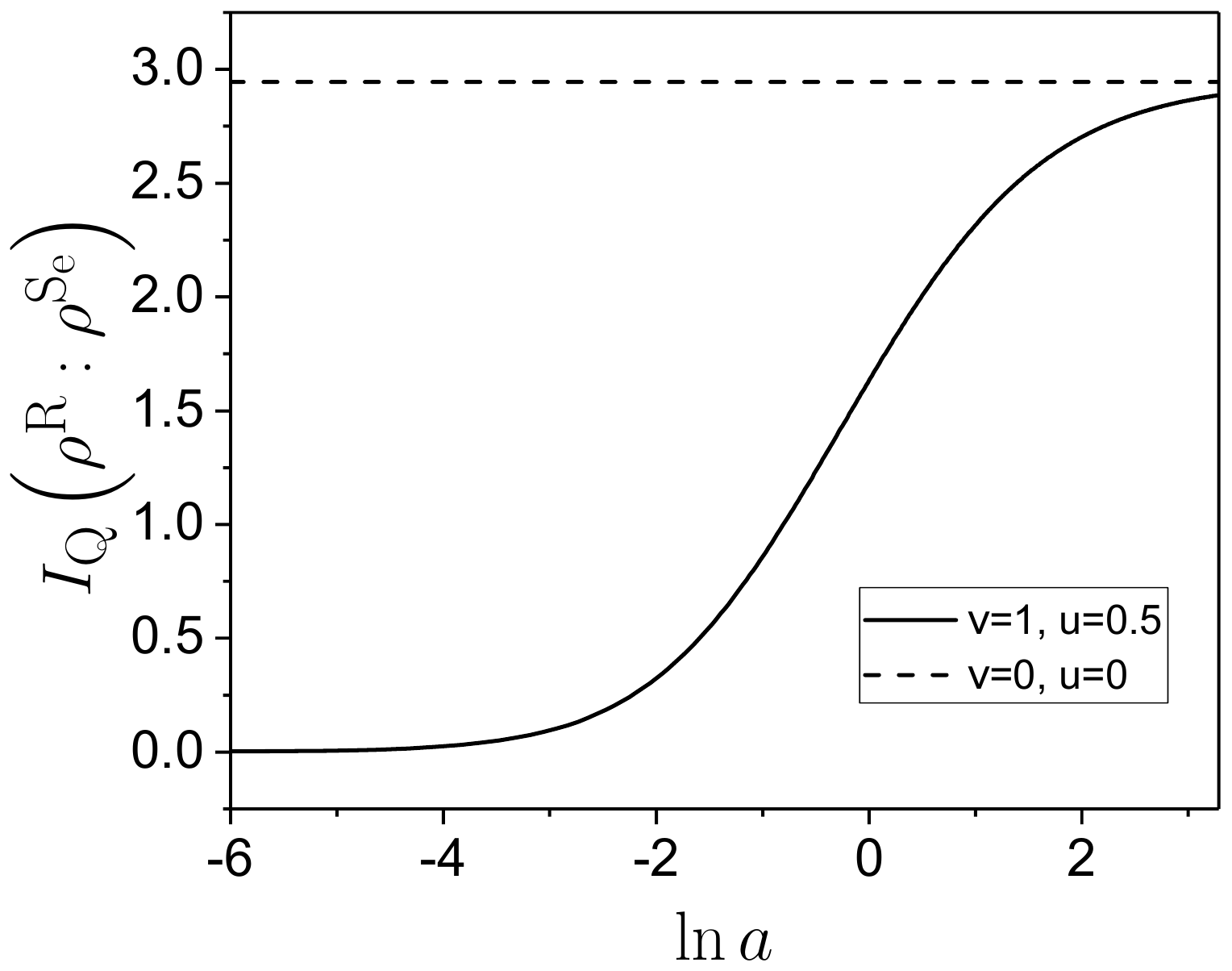}
\end{center}
\caption{Mutual information versus the squeezing parameter $\ln{a}$ for  one mode of a two mode squeezed vacuum state.}
\label{fig:RQSS}
\end{figure}
In~\S\ref{subsec:certificationtest}
we employ the CV quantum mutual-information approach to introduce a certification technique for CV RQSS schemes.

\subsection{Certification test for RQSS protocols}
\label{subsec:certificationtest}
In this subsection,  we establish our model for certification tests. Specifically, we introduce certification tests for~$\mathcal{A}$,~$\mathcal{F}$ and~$\mathcal{I}$, respectively.

\paragraph{RQSS certification for~$\mathcal{A}$.}
\label{para:RQSSCA}
Let~$I_\text{T}^\mathcal{A}$ be a threshold quantum mutual information chosen by the referee.
This quantum mutual information
quantifies the minimum knowledge that players in an access structure are able to obtain about the secret. Let $\beta>0$ be a maximum failure probability.
A test, which receives copies of some $X$ as input, and yields accept or reject, is a  test for certifying whether~$X \in \mathcal{A}$, if, with probability at least~$1-\beta$,
it both rejects every~$\rho^{\text{X}}$ for which
\begin{equation}\label{eq:soundnessqualified}
I\left(\text{X};\text{R}\right)<I_\text{T}^\mathcal{A}
\end{equation}
and accepts if
\begin{equation}
\label{eq:completenessqulified}
I\left(\text{X};\text{R}\right)\geq I_\text{T}^\mathcal{A}+\delta. 
\end{equation}
These conditions correspond to soundness~(\ref{eq:soundnessqualified}) and completeness~(\ref{eq:completenessqulified})
~\cite{kapourniotis2014verified,10.1007/978-3-540-24587-2_1,aolita2015reliable}.
\paragraph{RQSS certification for~$\mathcal{F}$.}
\label{para:RQSSCFFS}
Let $I_\text{T}^\mathcal{F}$ be a threshold quantum mutual information chosen by the referee,
which quantifies the maximum knowledge that players in the forbidden structure can obtain about the secret.
A test, which receives as input copies of some~$\rho^X$, and yields accept or reject, is a certification test for certifying whether $X \in \mathcal{F}$, if, with
probability at least $1-\beta$, it both accepts every $X$ for which 
\begin{equation}\label{eq:completenessforbiden}
I\left(\text{X};\text{R}\right)\leq I_\text{T}^\mathcal{F}-\delta,
\end{equation}
and rejects a different~$\rho^{\text{X}}$ for
\begin{equation}
\label{eq:sound1}
I\left(\text{X};\text{R}\right)>I_\text{T}^\mathcal{F}.
\end{equation}
These conditions are completeness~(\ref{eq:completenessforbiden}) and soundness~(\ref{eq:sound1}).
	
\paragraph{RQSS certification for~$\mathcal{I}$.}
\label{para:RQSSIS}
A test that receives copies of some~$X$ as input and yields accept or reject certifies whether~$X \in \mathcal{I}$ if, for a least probability~$1-\beta$,
it both rejects every $X$ for 
\begin{equation}\label{eq:soundnessintermediate1}
I\left(\text{X};\text{R}\right)\leq I_\text{T}^\mathcal{F}-\delta,
\end{equation}
or
\begin{equation}\label{eq:soundnessintermediate2}
I\left(\text{X};\text{R}\right)\geq I_\text{T}^\mathcal{A}+\delta. 
\end{equation}
and accepts if
\begin{equation}
\label{eq:completenessintermediae}
I_\text{T}^\mathcal{F}<I\left(\text{X};\text{R}\right)<I_\text{T}^\mathcal{A}.
\end{equation}
Conditions~(\ref{eq:soundnessintermediate1})
and~\ref{eq:soundnessintermediate2})
are soundness and
condition~(\ref{eq:completenessintermediae}) is
completeness. In the next subsection we employ our certification model to propose a practical test to ascertain RQSS protocols.

\subsection{Practical realization of the certification test}\label{sec:CP}
In this subsection, we propose a practical algorithm, for determining if~$X$ is in~$\mathcal{A}$,~$\mathcal{I}$ or~$\mathcal{F}$.
We prove propositions that the algorithm is both sound and complete. Furthermore, we provide a sufficiency test for the referee to know how many runs are required for her to have sufficient information to check if a particular element is in~$\mathcal{A}$,~$\mathcal{I}$ or~$\mathcal{F}$.
\subsubsection{Steps for certification}
Below we provide the steps for certifying RQSS.
Before commencing certification,
the referee numerically labels each element of the power set and proceeds to test each labelled element of the power set in order according to this labelling.
For simplicity,
and without loss of generality, we assume that each player holds one share; thus, the number~$n$ of modes
equals one more than the number of players, hence shares, in the given subset. This extra mode allows a single-mode reference field in addition to the modes held by the players.

The referee conducts a test that requires many rounds~(\ref{subsubsec:rounds}) for each power-set element. The test evaluates whether a fixed subset of players is in~$\mathcal{A}$,~$\mathcal{I}$ or~$\mathcal{F}$. In order to do so, the referee estimates the quantum mutual information~$I_\text{e}\left(\text{R},\text{S}_\text{e}\right)$ between the reference state~$\rho^\text{R}$ and the extracted secret state~$\rho^{\text{S}_\text{e}}$ such that 
\begin{equation}
    I_\text{e}\left(\text{R};\text{S}_\text{e}\right)\in\left[I\left(\text{R};\text{S}_\text{e}\right)-\epsilon,I\left(\text{R};\text{S}_\text{e}\right)+\epsilon\right],
\end{equation}
with a failure probability~$\beta<1/2$. Algorithm~\ref{alg:CertificationofRQSSprotocols} accepts~$I_\text{e}\left(\text{R},\text{S}_\text{e}\right)$ as input and determines the structure of the power-set element. If the test result is consistent with the dealer’s announcement that the access structure is the announced~$\Gamma$,
the referee announces a new subset of players to test; otherwise the procedure halts as the team of dealer and players has failed the certification test.

To estimate~$I_\text{e}\left(\text{R};\text{S}_\text{e}\right)$, the referee estimates the expectation values corresponding to each element of the matrices
\begin{align}
\label{eq:matrixG}
\bm{G}=&\begin{pmatrix}
2\hat{\bm{x}}^2_1&\frac{\left(\hat{\bm{x}}_1+\hat{\bm{x}}_2\right)^2}{2}&\hat{\bm{x}}_1\hat{\bm{x}}_3+\hat{\bm{x}}_3\hat{\bm{x}}_1&\hat{\bm{x}}_1\hat{\bm{x}}_4+\hat{\bm{x}}_4\hat{\bm{x}}_1\\
\frac{\left(\hat{\bm{x}}_1+\hat{\bm{x}}_2\right)^2}{2}&2\hat{\bm{x}}^2_2&\hat{\bm{x}}_2\hat{\bm{x}}_3+\hat{\bm{x}}_3\hat{\bm{x}}_2&\hat{\bm{x}}_2\hat{\bm{x}}_4+\hat{\bm{x}}_4\hat{\bm{x}}_2\\
\hat{\bm{x}}_1\hat{\bm{x}}_3+\hat{\bm{x}}_3\hat{\bm{x}}_1&\hat{\bm{x}}_2\hat{\bm{x}}_3+\hat{\bm{x}}_3\hat{\bm{x}}_2&2\hat{\bm{x}}^2_3&\frac{\left(\hat{\bm{x}}_3+\hat{\bm{x}}_4\right)^2}{2}\\
\hat{\bm{x}}_1\hat{\bm{x}}_4+\hat{\bm{x}}_4\hat{\bm{x}}_1&\hat{\bm{x}}_3\hat{\bm{x}}_4+\hat{\bm{x}}_4\hat{\bm{x}}_3&\frac{\left(\hat{\bm{x}}_3+\hat{\bm{x}}_4\right)^2}{2}&2\hat{\bm{x}}^2_4
\end{pmatrix},
\end{align}
and
\begin{equation}
\label{eq:matrixC}
\bm{C}=\begin{pmatrix}
\hat{\bm{x}}_1&\hat{\bm{x}}_2&
\hat{\bm{x}}_3&\hat{\bm{x}}_4
\end{pmatrix},
\end{equation}
with~$\hat{\bm{x}}$ defined in Eq.~(\ref{eq:eqthree}). The first and second modes hold reference and reconstructed secret states, respectively. The referee's result is then used to estimate the covariance matrix~(\ref{eq:eqone0}) of~$\rho^{\text{R}\text{S}_\text{e}}$ according to~\cite{aolita2015reliable}
\begin{align}
\label{eq:estimatecovariancematrix} 
V_{ij}^{\text{R}\text{S}_\text{e}}&=\langle\bm{G}_{ij}\rangle-\langle\bm{C}_i\rangle\langle\bm{C}_{j}\rangle, \qquad ij\notin \{12,21,34,43\},\\
V_{ij}^{\text{R}\text{S}_\text{e}}&=2\langle\bm{G}_{ij}\rangle-\langle\bm{G}_{ii}\rangle/2-\langle\bm{G}_{jj}\rangle/2-\langle\bm{C}_i\rangle\langle\bm{C}_{j}\rangle,\qquad ij\in \{12,21,34,43\}.
\end{align}
This covariance matrix is used to calculate the entropies of~$\rho^{\text{S}_\text{e}},\rho^{\text{R}}$ and~$\rho^{\text{R}\text{S}_\text{e}}$ using Algorithm~\ref{alg:CVQENTROPY}. The resultant entropies are then inserted into the standard formula for quantum mutual information~(\ref{eq:RMIH}).

The expectation value of each element of~(\ref{eq:matrixG}) and~(\ref{eq:matrixC}) is calculated by performing multiple homodyne measurements on identical and independent copies of~$\rho^{\text{R}\text{S}_\text{e}}$ and taking the average of the measurement results.
Using Chebyshev's inequality~\cite{aolita2015reliable},
the referee calculates an upper-bound for the estimation error of each expectation value as a function of number of rounds and~$\beta$.
Subsequently, this estimation error is then used to calculate the maximum expectation values' estimation error~$\epsilon_\text{max}$ of covariance-matrix entries via the standard formula for error propagation.
Afterwards she calculates the bound on the estimation error of entropies following Algorithm~\ref{alg:Upperboundofestimationerrorofquantummutualinformation}. The estimation error of~$I_\text{e}\left(\text{R};\text{S}_\text{e}\right)$ is bounded by summation of the entropies estimation errors. The rounds continue until the estimation error of~$I_\text{e}\left(\text{R};\text{S}_\text{e}\right)$ is below a prespecified acceptable~$\epsilon$ error.

\begin{algorithm}[H]
\caption{Continuous-variable quantum entropy~($H_\text{vN}$).}
\label{alg:CVQENTROPY}
\begin{algorithmic}
\Require{
\Statex $n\in\mathbb{N}$ \Comment{Number of modes}
\Statex $\bm{V}\in \mathbb{R}^{2n}\times \mathbb{R}^{2n}$ \Comment{Covariance matrix}
\Statex $\bm{\Omega}\in\mathbb{Z}^{2n}\times\mathbb{Z}^{2n}~(\ref{eq:eqfour})$
 }
\Ensure{
\Statex $H_\text{vN}\in \mathbb{R}^+$ \Comment{von Neumann entropy}
}
\Function{vonNeumannH}{$\bm{V}$}
\State $\bm{\nu}$$\leftarrow\text{Eigenvalues}_+\left(\text{i}\Omega \bm{V}\right)$.\Comment{Calculates positive eigenvalues.}
\Statex $\bm{\nu}^\pm\leftarrow\frac{\bm{\nu}\pm\mathds{1}}{2}$.\\
\Return $H_\text{vN}
\leftarrow\bm{\nu}^+\cdot\log\bm{\nu}^+
+\bm{\nu}^-\cdot\log\bm{\nu}^-.$
\EndFunction
\end{algorithmic}
\end{algorithm}

\begin{algorithm}
\caption{Upper bound of $H_\text{vN}$ estimation error.}
\label{alg:Upperboundofestimationerrorofquantummutualinformation}
\begin{algorithmic}
\Require{
\Statex $n\in\mathbb{N}$ \Comment{Number of modes}
\Statex $\bm{V}\in \mathbb{R}^{2n}\times \mathbb{R}^{2n}$ \Comment{Covariance matrix}
\Statex  $\epsilon_\text{max}$ \Comment{Maximum estimation error of covariance matrix elements}
}
\Ensure{
\Statex $H_{\text{vN,error}}^\text{upper}\in \mathbb{R}^+$ \Comment{Upper bound of QMI estimation error} 
}
 \Function{$H_{\text{vN},\text{error}}^\text{upper}$}{$\bm{V},\epsilon_\text{max}$}
\label{alg:qmierror}
\State $\sigma_\text{max} \gets$ maximal singular value of $\bm{V}$.
\State $\sigma_\text{min} \leftarrow$ minimal singular value of $\bm{V}$.\\
\Return  $H_{\text{vN},\text{error}}^\text{upper}$$\leftarrow\kappa\left(1+\log\left(2n\sigma_\text{max}\right)\right)2n\epsilon_\text{max}.$\Comment{$\kappa=\frac{\sigma_\text{max}}{\sigma_\text{min}}$ is always finite.}
\EndFunction
\end{algorithmic}
\end{algorithm}

\begin{algorithm}
\caption{Estimation of QMI.}
\label{alg:Estimation of quantum mutual information}
\begin{algorithmic}
\Require{
\Statex $T\in\mathbb{N}$ \Comment{Number of trials}
\Statex $\rho^{\otimes T}\in \mathcal{B}\left(L^2(\mathbb{R}^{2T})\right)$ \Comment{$T$ copies of the joint state~$\rho$ for the reference and players' reconstructed state}
\Statex  $\epsilon\in \mathbb{R}^+$
\Comment{Error tolerance for estimated QMI}
\Statex $\textsc{Tol}\in\interval({0,1/2})$
\Comment{Failure probability tolerance}
\Statex $\sigma \in \mathbb{R}^+$ \Comment{A uniform upper bound on the standard deviations of measurement results}
\Statex $\textsc{HomMeas}[\rho,x,\textsc{Mode},\theta]$ \Comment{Homodyne measurement on mode~\textsc{Mode}$\in\{0,1\}$
with respect to local-oscillator phase~$\theta$; replaces~$\rho$ by some~$\ket{x}\bra{x}$ with probability~$\bra{x}\rho\ket{x}$}
}
\Ensure{
\Statex $\textsc{EstQMI}\in \mathbb{R}^+$ \Comment{Estimated QMI}
}
\Procedure{EstimateQMI}{$\epsilon$,\textsc{Tol},$T,\rho^{\otimes T},\sigma,\textsc{HomMeas}[\rho,x,\textsc{Mode},\theta]$}
\For{$i$ from~$1$ to~$2$}
\For{$j$ from~$1$ to~$2$}
\State $\textsc{CovRecon}\left[ij\right]\gets 0$  \Comment{Initialize covariance matrix for the players' reconstructed state including position-position, position-momentum, momentum-position and momentum-momentum}
\State $\textsc{CovRef}\left[ij\right]\gets 0$  \Comment{Initialize covariance matrix for the reference state including position-position, position-momentum, momentum-position and momentum-momentum}
\EndFor
\EndFor
\For{$i$ from~$1$ to~$4$}
\State $\textsc{HomResult}\left[i\right]\gets 0$
\Comment{Initialize vector comprising sums of in-phase
    and out-of-phase homodyne measurements of modes~0 and~1}
\For{$j$ from~$1$ to~$4$}
\State $\textsc{CovRecRef}\left[ij\right]\gets 0$  \Comment{Initialize joint reconstructed-reference covariance matrix including position-position, position-momentum, momentum-position and momentum-momentum}
\State \textsc{SecondMom}$\left[ij\right]\gets 0$
\Comment{Second-moment matrix defined in Eq.~(\ref{eq:matrixG})}
\EndFor
\EndFor
\State $\varepsilon \gets \left\lceil\sigma\sqrt{\frac{1}{l\left(1-\left(1-\textsc{Tol}\right)^{1/14}\right)}}\right\rceil$\Comment{Maximum estimation error of measurement results expectation values with a least probability~$\textsc{Tol}$}
\algstore{myalg}
\end{algorithmic}
\end{algorithm}
\begin{algorithm}                     
\begin{algorithmic}           
\algrestore{myalg}
\State $l\gets0$\Comment{Number of times that the referee performs the sufficiency test}
\State $\textsc{Rho}\gets\rho$ \Comment{Initialize $\textsc{Rho}$ to the first of input $\rho^{\otimes T}$}
\State $\epsilon_\text{QMI}\gets 2\epsilon$
\Comment{Initialize to any value greater than~$\epsilon$}
\For{$r$ from~$1$ to~$T$}
\While{$\epsilon_\text{QMI}>\epsilon$}
\State $l\gets l+1$
\Comment{Increment the sufficiency-test counter}
\If{$14l>T$}
\Comment{Referee measures 14 copies before ascertaining sufficiency}
\Return Fail
\State \textsc{Exit}
\Comment{Abort procedure if fewer than 14 copies remain}
\EndIf
\If{$r-1$ \textsc{mod} $14$=0}    \Comment{Measure one of~$T$ copies of~$\rho$}
\State Call \textsc{HomMeas}$(\textsc{Rho},x,0,0)$
\Comment{In-phase homodyne measurement of the reconstructed state}
\State $\textsc{HomResult}[1]\gets\textsc{HomResult}[1]+x$
\Comment{Sum detection outcomes}
\ElsIf{$r-2$ \textsc{mod} $14$=0 }\Comment{Measure one of~$T$ copies of~$\rho$}
\State Call \textsc{HomMeas}$(\textsc{Rho},x,0,\frac{\pi}{2})$ \Comment{Out-of-phase homodyne measurement of the reconstructed state}
\State $\textsc{HomResult}[2]\gets \textsc{HomResult}[2]+x$\Comment{Sum detection outcomes}
\ElsIf{$r-3$ \textsc{mod} $14$=0}\Comment{Measure one of~$T$ copies of~$\rho$}
\State Call \textsc{HomMeas}$(\textsc{Rho},x,1,0)$ \Comment{In-phase homodyne measurement of the reference state}
\State $\textsc{HomResult}[3]\gets \textsc{HomResult}[3]+x$\Comment{Sum detection outcomes}
\ElsIf{$r-4$ \textsc{mod} $14$=0}\Comment{Measure one of~$T$ copies of~$\rho$}
\State Call \textsc{HomMeas}$(\textsc{Rho},x,1,\frac{\pi}{2})$ \Comment{Out-of-phase homodyne measurement of the reference state}
\State $\textsc{HomResult}[4]\gets \textsc{HomResult}[4]+x$\Comment{Sum detection outcomes}
\ElsIf{$r-5$ \textsc{mod} $14$=0}\Comment{Measure one of~$T$ copies of~$\rho$}
\State Call \textsc{HomMeas}$(\textsc{Rho},x,0,0)$ \Comment{In-phase homodyne measurement of the reconstructed state}
\State $\textsc{SecondMom}[11]\gets\textsc{SecondMom}[11]+2x^2$
\ElsIf{$r-6$ \textsc{mod} $14$=0}\Comment{Measure one of~$T$ copies of~$\rho$}
\State Call \textsc{HomMeas}$(\textsc{Rho},x,0,0)$ \Comment{In-phase homodyne measurement of the reconstructed state}
\State $y \gets x$
\State Call \textsc{HomMeas}$(\textsc{Rho},x,1,0)$ \Comment{In-phase homodyne measurement of the reference state}
\algstore{myalg}
\end{algorithmic}
\end{algorithm}
\begin{algorithm}                     
\begin{algorithmic}           
\algrestore{myalg}
\State $\textsc{SecondMom}[13]\gets\textsc{SecondMom}[13]+2xy$
\State $\textsc{SecondMom}[31]\gets\textsc{SecondMom}[13]$
\ElsIf{$r-7$ \textsc{mod} $14$=0}\Comment{Measure one of~$T$ copies of~$\rho$}
\State Call \textsc{HomMeas}$(\textsc{Rho},x,0,0)$ \Comment{In-phase homodyne measurement of the reconstructed state}
\State $y \gets x$
\State Call \textsc{HomMeas}$(\textsc{Rho},x,1,\frac{\pi}{2})$ \Comment{Out-of-phase homodyne measurement of the reference state}
\State $\textsc{SecondMom}[14]\gets\textsc{SecondMom}[14]+2xy$
\State $\textsc{SecondMom}[41]\gets\textsc{SecondMom}[14]$
\ElsIf{$r-8$ \textsc{mod} $14$=0} \Comment{Measure one of~$T$ copies of~$\rho$}
\State Call \textsc{HomMeas}$(\textsc{Rho},x,0,\frac{\pi}{2})$ \Comment{Out-of-phase homodyne measurement of the reconstructed state}
\State $\textsc{SecondMom}[22]\gets\textsc{SecondMom}[22]+2x^2$
\ElsIf{$r-9$ \textsc{mod} $14$=0}\Comment{Measure one of~$T$ copies of~$\rho$}
\State Call \textsc{HomMeas}$(\textsc{Rho},x,0,\frac{\pi}{2})$\Comment{Out-of-phase homodyne measurement of the reconstructed state} 
\State $y\gets x$
\State Call \textsc{HomMeas}$(\textsc{Rho},x,1,0)$\Comment{In-phase homodyne measurement of the reference state}
\State $\textsc{SecondMom}[23]\gets\textsc{SecondMom}[23]+2xy$
\State $\textsc{SecondMom}[32]\gets\textsc{SecondMom}[23]$
\ElsIf{$r-10$ \textsc{mod} $14$=0}\Comment{Measure one of~$T$ copies of~$\rho$}
\State Call \textsc{HomMeas}$(\textsc{Rho},x,0,\frac{\pi}{2})$ \Comment{Out-of-phase homodyne measurement of the reconstructed state}
\State $y\gets x$
\State Call \textsc{HomMeas}$(\textsc{Rho},x,1,\frac{\pi}{2})$ \Comment{Out-of-phase homodyne measurement of the reference state}
\State $\textsc{SecondMom}[24]\gets\textsc{SecondMom}[24]+2xy$
\State $\textsc{SecondMom}[42]\gets\textsc{SecondMom}[24]$
\ElsIf{$r-11$ \textsc{mod} $14$=0}\Comment{Measure one of~$T$ copies of~$\rho$}
\State Call \textsc{HomMeas}$(\textsc{Rho},x,1,0)$ \Comment{In-phase homodyne measurement of the reference state}
\State $\textsc{SecondMom}[33]\gets\textsc{SecondMom}[33]+2x^2$
\ElsIf{$r-12$ \textsc{mod} $14$=0}\Comment{Measure one of~$T$ copies of~$\rho$}
\State Call \textsc{HomMeas}$(\textsc{Rho},x,1,\frac{\pi}{2})$ \Comment{Out-of-phase homodyne measurement of the reference state}
\State $\textsc{SecondMom}[44]\gets\textsc{SecondMom}[44]+2x^2$
\ElsIf{$r-13$ \textsc{mod} $14$=0 }\Comment{Measure one of~$T$ copies of~$\rho$}
\algstore{myalg}
\end{algorithmic}
\end{algorithm}
\begin{algorithm}                     
\begin{algorithmic}           
\algrestore{myalg}
\State Call \textsc{HomMeas}$(\textsc{Rho},x,0,\frac{\pi}{4})$\Comment{Homodyne measurement of the reconstructed state with respect to local-oscillator phase~$\frac{\pi}{4}$} 
\State $\textsc{SecondMom}[12]=2x^2-\textsc{SecondMom}[11]^2-\textsc{SecondMom}[22]^2$
\State $\textsc{SecondMom}[21]\gets\textsc{SecondMom}[12]$
\Else{$r-14$ \textsc{mod} $14$=0 }\Comment{Measure one of~$T$ copies of~$\rho$}
\State Call \textsc{HomMeas}$(\textsc{Rho},x,1,\frac{\pi}{4})$ \Comment{Homodyne measurement of the reference state with respect to local-oscillator phase~$\frac{\pi}{4}$}
\State $\textsc{SecondMom}[34]=2x^2-\textsc{SecondMom}[33]^2-\textsc{SecondMom}[44]^2$
\State $\textsc{SecondMom}[43]\gets\textsc{SecondMom}[34]$
\EndIf
\For{$i$ from~$1$ to~$4$}
\For{$j$ from~$i$ to~$4$}
\State $\textsc{CovRecRef}\left[ij\right]\gets \frac{1}{l}\left(\textsc{SecondMom}[ij]-\textsc{HomResult}[i]
    \textsc{HomResult}[j]\right)$  
\State $\textsc{CovRecRef}\left[ij\right]\gets \textsc{CovRecRef}\left[ji\right]$
\EndFor
\EndFor
\For{$i$ from~$1$ to~$2$}
\For{$j$ from~$1$ to~$2$}
\State $\textsc{CovRecon}\left[ij\right]\gets \textsc{CovRecon}\left[ij\right]$
\State $\textsc{CovRef}\left[ij\right]\gets \textsc{CovRecon}\left[i+2j+2\right]$
\EndFor
\EndFor
\begin{equation}
\label{eq=epsilonmax1}
\varepsilon_\text{max}\gets\frac{\varepsilon}{l}
\text{max}_{ij}\sqrt{1+\left(\textsc{HomResult}\left[i\right]\right)^2+\left(\textsc{HomResult}\left[j\right]\right)^2}
\end{equation}
\begin{align}
    \label{eq:epsilonmax}
  \varepsilon_\text{max}  \gets&\text{max}\bigg\{\varepsilon_\text{max},\frac{\varepsilon}{l}\sqrt{4+\left(\textsc{HomResult}\left[1\right]\right)^2+\left(\textsc{HomResult}\left[2\right]\right)^2},\\
  &\frac{\varepsilon}{l}\sqrt{4+\left(\textsc{HomResult}\left[3\right]\right)^2+\left(\textsc{HomResult}\left[4\right]\right)^2}\bigg\}
\end{align}\Comment{Via standard error propagation method}

\State $\epsilon_\text{QMI}\gets\sum_{\text{Q=Rp,p,R}}\rm{H_{\text{vN},\text{error}}^\text{upper}}\big(\bm{V}^\text{e,\text{Q}},$$\varepsilon_\text{max}\big)$\Comment{See Algorithm~\ref{alg:Upperboundofestimationerrorofquantummutualinformation}}
\EndWhile
\EndFor
\Return $\textsc{EstQMI}\gets\sum_{\text{Q=R,p}}\rm{vonNeumannH}(\bm{V}^\text{e,\text{Q}})-\rm{vonNeumannH}(\bm{V}^\text{e,\text{Rp}})$\Comment{see Algorithm~\ref{alg:CVQENTROPY}}
\EndProcedure
\end{algorithmic}
\end{algorithm}
 \clearpage

\begin{algorithm}[H]
\caption{Certification of RQSS protocols.}
\label{alg:CertificationofRQSSprotocols}
\begin{algorithmic}
\Require{
\Statex $T\in\mathbb{N}$ \Comment{Number of trials for each instance}
\Statex $I_\text{T}^\mathcal{F}\in \mathbb{R}^+$ \Comment{Threshold quantum mutual information for the forbidden structure}
\Statex  $I_\text{T}^\mathcal{A}\in \mathbb{R}^+$ \Comment{Threshold quantum mutual information for all authorized structures}
\Statex  $\epsilon\in \mathbb{R}^+$ \Comment{Estimation error bound of estimated QMI}
\Statex $\textsc{Tol}\in\interval({0,1/2})$
\Comment{Maximum failure probability}
\Statex  $P\in \mathbb{N}$ \Comment{Cardinality of the set of players}
\Statex $\textsc{F}\left[J\right]\in\{0,1,2\}$\Comment{Returns~$J^\text{th}$ power set of players structure claimed by the dealer~(\ref{eq:eq64})}
\Statex $\bigotimes_{J=1}^{2^P-1}\rho_{J}^{\otimes T}\in \mathcal{B}\left(L^2(\mathbb{R}^{2^PT})\right)$\Comment{$\rho_{J}$ is the joint state for the reference and players' reconstructed state for~$J^\text{th}$ subset of players}
\Statex $\sigma \in \mathbb{R}^+$ \Comment{A uniform upper bound on the standard deviations of measurement results}
\Statex $\textsc{HomMeas}[\rho,x,\textsc{Mode},\theta]$ \Comment{Homodyne measurement on mode~\textsc{Mode}$\in\{0,1\}$
with respect to local-oscillator phase~$\theta$; replaces~$\rho$ by some~$\ket{x}\bra{x}$ with probability~$\bra{x}\rho\ket{x}$}
}
\Ensure{
\Statex   $b\in\{0,1\}$\Comment{Certify ($b=1$) or not certify ($b=0$)}
}
\Procedure{Certification}{$I_\text{T}^\mathcal{F},I_\text{T}^\mathcal{A},\epsilon,P,\bigotimes_{J=1}^{2^P-1} \rho_{J}^{\otimes T},\textsc{F}[J],\sigma,\textsc{Tol},\textsc{HomMeas}[\rho,x,\textsc{Mode},\theta]$}
\State $c\gets \textsc{F}\left[1\right]$\Comment{initialize the structure of power-set elements based on referees' test to~$\textsc{F}\left[1\right]$}
\State $\textsc{pass}\gets 0$ \Comment{initialize the number of power-set elements that pass the test}
\For{$J$ from $1$ to $2^P-1$}
\State $\textsc{EstQMI}\gets\textsc{EstimateQMI}\left(\epsilon,\textsc{Tol},T,\rho^{\otimes T}_J,\sigma,\textsc{HomMeas}[\rho,x,\textsc{Mode},\theta]\right)$\Comment{see Algorithm~\ref{alg:Estimation of quantum mutual information}.}
\If{\begin{equation} 
\label{eq:eq77}
\textsc{EstQMI}>I_\text{T}^\mathcal{A}+\epsilon,
\end{equation}}
\State $c\leftarrow2$
\ElsIf{\begin{equation}
\label{eq:eq79}
I_\text{T}^\mathcal{F}-\epsilon<\textsc{EstQMI}<I_\text{T}^\mathcal{A}+\epsilon, 
\end{equation}}
\algstore{myalg}
\end{algorithmic}
\end{algorithm}
\begin{algorithm}                     
\begin{algorithmic}           
\algrestore{myalg}
\State $c\leftarrow1$
\Else{}
\State $c\leftarrow 0$
\EndIf
\If{$c=\textsc{f}[J]$}
$\textsc{pass}\gets\textsc{pass}+1$
\Else \State \textsc{Exit}
\Comment{Halt}
\EndIf
\EndFor
\If{$\textsc{pass}=2^P$}
\State  $b\gets 1$.
\Else
\State  $b\gets 0$.
\EndIf
\State \Return $b$
\EndProcedure
\end{algorithmic}
\end{algorithm}
\clearpage
\begin{prop}
Algorithm~\ref{alg:Estimation of quantum mutual information} ensures
\begin{equation}
\label{eq:eq80}
\text{pr}\left[\left|I_\text{e}\left(X;R\right)-I\left(X;R\right)\right|\leq \epsilon_\text{QMI}\right]\geq1-\beta,
\end{equation}
and
\begin{equation}
	\epsilon_\text{QMI}\in O\left(\frac{1}{\sqrt{N}}\right)
\end{equation}
for~$N$ the number of rounds.
\end{prop}
\begin{proof}
Using Chebyshev's inequality~\cite{aolita2015reliable},
\begin{align}
\label{eq:eq85}
	\text{pr}\left[\left|\bar{\bm{G}}_{ij}
		-\mathbb{E}\left(\bm{G}_{ij}\right)\right|
			\geq \epsilon\right]
	\leq&\frac{\sigma^2}{\epsilon^2l},\\
\label{eq:eq140}
\text{pr}\left[\left|\bar{\bm{C}}_i-\mathbb{E}\left(\bm{C}_i\right)\right|\geq \epsilon\right]
	\leq&\frac{\sigma^2}{\epsilon^2l}.
\end{align}
Equations~(\ref{eq:eq85}) and~(\ref{eq:eq140}) equivalently are
\begin{align}
\label{eq:eq170}
	\text{pr}\left[\left|\bar{\bm{G}}_{ij}
		-\mathbb{E}\left(\bm{G}_{ij}\right)\right|
			\leq\epsilon\right]
		\geq&1-\frac{\sigma^2}{\epsilon^2l},
	\\
\label{eq:eq171}
\text{pr}\left[\left|\bar{\bm{C}}_i
	-\mathbb{E}\left(\bm{C}_i\right)\right|
		\leq\epsilon\right]\geq&1-\frac{\sigma^2}{\epsilon^2l}.
\end{align}
Assigning
\begin{equation}
\label{eq:eq89}
	\epsilon
		\gets
			\left\lceil\frac{\sigma}
				{\sqrt{l\left(1-\left(1-\beta\right)^\frac{1}{14}\right)}}\right\rceil
\end{equation}
and assuming an independent identically distributed (iid) protocol delivers 
\begin{equation}
\text{pr}\left[\forall i,j:\left|\bar{\bm{C}}_i-\mathbb{E}\left(\bm{C}_i\right)\right|\wedge \left|\bar{\bm{G}}_{ij}-\mathbb{E}\left(\bm{G}_{ij}\right)\right|\leq \epsilon\right]\geq1-\beta.
\end{equation}
Let~$\epsilon_\text{max}$ be the maximum estimation error of estimated covariance matrix, which is calculated in terms of~$\epsilon$ (\ref{eq:eq89}) via standard error propagation methods. In the following we give an upper bound on the estimation error of quantum mutual information in terms of~$\epsilon_\text{max}$. In order to do so, we introduce some helpful notation and theorems used in our proofs.

For any two Gaussian states with corresponding covariance matrices  $\bm{V}_\text{A}$ and $\bm{V}_\text{B}$, the entropy difference is bounded by~\cite{IDEL201745}
\begin{equation}
\label{eq:eq93}
\left|H_\text{vN}\left(\bm V_\text{A}\right)-H_\text{vN}\left(\bm V_\text{B}\right)\right|\leq \kappa\left(\bm V_\text{A}\right)K\|\bm{V}_\text{A}-\bm{V}_\text{B}\|_1,
\end{equation}
for
\begin{equation}
\label{eq:eq94}
	K:=1+\log\left[\text{max}\left(\|\bm V_\text{A}\|_\infty,\frac{1}{2}\left(\|\bm V_\text{A}^{-1}\|_\infty^{-1}-1\right)\right)\right].
\end{equation}
Also 
\begin{equation}
\label{eq:eq95}
\norm{\bm{A}^{-1}}_\infty^{-1}\leq \norm{\bm{A}}_\infty,
\end{equation}
holds for any covariance matrix $\bm{A}$~\cite{GRCAR2010203}.
Hence,
\begin{equation}
\label{eq:eq96}
\frac{1}{2}\left(\norm{\bm A^{-1}}_\infty^{-1}-1\right)\leq \norm{\bm A}_\infty.
\end{equation}
By substituting~Eq.~(\ref{eq:eq96}) into~Eq.~(\ref{eq:eq94}),
we obtain the perturbation bound
\begin{equation}
\label{eq:eq97}
\left|H_\text{vN}\left(\bm V_\text{A}\right)-H_\text{vN}\left(\bm V_\text{B}\right)\right|\leq \kappa\left(\bm V_\text{A}\right)\|\bm{V}_\text{A}-\bm{V}_\text{B}\|_1\left(1+\log\left(\|\bm V_\text{A}\|_\infty\right) \right).
\end{equation}
For any $\text{Q}\in \{\text{R},\text{P},\text{RP}\}$, let~$\bm V^\text{e,Q}$ and~$\bm V^\text{Q}$ be the estimated and real covariance matrices, respectively. Then
\begin{equation}
\label{eq:eq98}
\norm{\bm V^\text{e,Q}}_\infty \leq \norm{\bm{U}}_\infty \norm{\bm{\Sigma}}_\infty \norm{\bm{V}}_\infty \leq \sigma_\text{max,e,Q}\dim{\bm V^\text{e,Q}}.
\end{equation}
Also
\begin{equation}
\label{eq:eq99}
\norm{\bm{V}^\text{Q}-\bm{V}^\text{e,Q}}_1\leq \varepsilon_\text{max}\dim{\bm V^\text{e,Q}}.
\end{equation}
Furthermore, let us define
\begin{equation}
\label{eq:eq100}
	\Delta{H_\text{vN}}\left(\text{Q}\right)
		:=H_\text{vN}\left(\bm{V}^\text{Q}\right)-H_\text{vN}\left(\bm{V}^\text{e,Q}\right),
\end{equation}
and
\begin{equation}
\label{eq:eq101}
    \Delta I\left(X\right)=I\left(X;R\right)-I_\text{e}\left(X;R\right).
\end{equation}
Thus,
\begin{equation}
\label{eq:102}
\Delta I\left(X\right)=\Delta{H_\text{vN}}\left(\text{X}\right)+\Delta{H_\text{vN}}\left(\text{R}\right)-\Delta{H_\text{vN}}\left(\text{RX}\right)
\end{equation}
Due to the triangle inequality,
\begin{align}
\label{eq:103}
\left|\Delta I\left(X\right)\right|\leq \left|\Delta{H_\text{vN}}\left(\text{X}\right)\right|+\left|\Delta{H_\text{vN}}\left(\text{R}\right)\right|+\left|\Delta{H_\text{vN}}\left(\text{RX}\right)\right|.
\end{align}
Each of the terms in the right-hand side of Eq.~(\ref{eq:103}) is suitably achieved by using Eq.~(\ref{eq:eq97}). Substituting Eqs.~(\ref{eq:eq99}) and~(\ref{eq:eq100}) into the resultant equation delivers Eq.~(\ref{eq:eq80}).

Now we show that~$\epsilon_\text{QMI}$ scales properly with respect to number of rounds. Using the Weyl~\cite{Weyl1912} perturbation bound for singular value decomposition, we conclude
\begin{equation}
	\kappa\left(\bm{V}^\text{e,Q}\right),\sigma_\text{max,e,Q}\in O(1),\; 
	\varepsilon_\text{max}\in O\left(\frac{1}{\sqrt{N}}\right).
\end{equation}
Therefore, the error bound scales inversely with square root of the number of rounds. Next we prove the algorithm~\ref{alg:CertificationofRQSSprotocols} is both sound and complete.
\end{proof}

\begin{prop}
\begin{enumerate}[(i)]
\item If $X\in \mathcal{A}$, Algorithm~\ref{alg:CertificationofRQSSprotocols} passes with probability at least $1-\beta$ and
\item	if $X\notin \mathcal{A}$ then the algorithm fails with probability at least $1-\beta$.
\end{enumerate}
\end{prop}
\begin{proof}
We show cases~(i) and~(ii) in sequence.\\
Case~(i):
We first recall that
\begin{equation}
\label{eq:106}
X\in \mathcal{A}\implies I\left(X;R\right)\geq I_\text{T}^\mathcal{A}+\delta.
\end{equation}
Also
\begin{equation}
\label{eq:eq107}
\text{pr}\left[\left|I\left(X;R\right)-I_\text{e}\left(X;R\right)\right|\leq\epsilon\right]\geq1-\beta.
\end{equation}
Therefore,
\begin{equation}
\label{eq:108}
\text{pr}\left[I_\text{T}^\mathcal{A}+\delta-\epsilon\leq I_\text{e}\left(X;R\right)\right]\geq1-\beta.
\end{equation}
As~$\delta-\epsilon\geq \epsilon$, we conclude
\begin{equation}
\label{eq:109}
\text{pr}\left[I_\text{T}^\mathcal{A}+\epsilon\leq I_\text{e}\left(X;R\right)\right]\geq1-\beta.
\end{equation}
Thus, Algorithm~\ref{alg:CertificationofRQSSprotocols} accepts with probability 
at least $1-\beta$ if $X\in \mathcal{A}$.\\
Case~(ii):
We note that
\begin{equation}\label{eq:eq110}
\text{pr}\left[ I_\text{e}\left(X;R\right)-\epsilon\leq I\left(X;R\right)\right]\geq1-\beta
\end{equation}
Therefore, substituting Eq.~(\ref{eq:soundnessqualified}) into Eq.~(\ref{eq:eq110}) delivers
\begin{equation}
\label{eq:111}
\text{pr}\left[I_\text{e}\left(X;R\right)<I_\text{T}^\mathcal{A}+\epsilon\right]\geq1-\beta.
\end{equation}
Thus, Algorithm~\ref{alg:CertificationofRQSSprotocols} rejects with probability at least $1-\beta$ if $X\notin \mathcal{A}$.
\end{proof}
\begin{prop}
\begin{enumerate}[(i)]
\item 	If $X\in \mathcal{F}$, then Algorithm~\ref{alg:CertificationofRQSSprotocols} accepts with probability at least $1-\beta$ and
\item	if $X\notin \mathcal{F}$ then Algorithm~\ref{alg:CertificationofRQSSprotocols} rejects with probability at least $1-\beta$.
\end{enumerate}
\end{prop}
\begin{proof}
We show cases~(i) and~(ii) in sequence.\\
Case~(i):
\begin{equation}
\label{eq:eq112}
	X\in \mathcal{F}\implies I\left(X;R\right)\leq I_\text{T}^\mathcal{F}-\delta.
\end{equation}
Also
\begin{equation}
\label{eq:eq113}
\text{pr}\left[\left|I\left(X;R\right)-I_\text{e}\left(X;R\right)\right|\leq\epsilon\right]\geq1-\beta.
\end{equation}
Therefore,
\begin{equation}
\label{eq:eq114}
\text{pr}\left[I_\text{e}\left(X;R\right)\leq I\left(X;R\right)+\epsilon\right]\geq1-\beta.
\end{equation}
Substituting Eq.~(\ref{eq:eq112}) in Eq.~(\ref{eq:eq114}) delivers
\begin{equation}
\label{eq:eq115}
\text{pr}\left[I_\text{e}\left(X;R\right)\leq I_\text{T}^\mathcal{F}-\delta+\epsilon\right]\geq1-\beta.
\end{equation}
As~$\delta-\epsilon\geq \epsilon$, we conclude
\begin{equation}
\label{eq:eq116}
\text{pr}\left[I_\text{e}\left(X;R\right)\leq I_\text{T}^\mathcal{F}-\epsilon\right]\geq1-\beta.
\end{equation}
Thus, Algorithm~\ref{alg:CertificationofRQSSprotocols} accepts with probability at least $1-\beta$ if $X\in \mathcal{F}$.\\
Case (ii):
\begin{equation}
\label{eq:eq117}
\text{pr}\left[I\left(X;R\right)-\epsilon\leq I_\text{e}\left(X;R\right)\right]\geq1-\beta.
\end{equation}
Substituting Eq.~(\ref{eq:sound1}) into Eq.~(\ref{eq:eq114}) delivers
\begin{equation}
\label{eq:eq118}
\text{pr}\left[I_\text{T}^\mathcal{F}-\epsilon\leq I_\text{e}\left(X;R\right)\right]\geq1-\beta.
\end{equation}
Thus, Algorithm~\ref{alg:CertificationofRQSSprotocols} rejects with probability at least $1-\beta$ if $X\notin \mathcal{F}$.
\end{proof}
\begin{prop}
\begin{enumerate}[(i)]
\item 	If $X\in \mathcal{I}$, then Algorithm~\ref{alg:CertificationofRQSSprotocols} accepts with probability at least $1-\beta$ and
\item	$X\notin \mathcal{I}$ then Algorithm~\ref{alg:CertificationofRQSSprotocols} rejects with probability at least $1-\beta$.
\end{enumerate}
\end{prop}
\begin{proof}
We show cases~(i) and~(ii) in sequence.\\
Case~(i):
\begin{equation}
\label{eq:eq119}
	X\in \mathcal{I}\implies I_\text{T}^\mathcal{F}<I\left(X;R\right)<I_\text{T}^\mathcal{A}.
\end{equation}
Also
\begin{equation}
\label{eq:eq120}
\text{pr}\left[\left|I\left(X;R\right)-I_\text{e}\left(X;R\right)\right|\leq\epsilon\right]\geq1-\beta.
\end{equation}
Therefore,
\begin{align}
\label{eq:eq121}
\text{pr}\left[ I\left(\text{X};\text{R}\right)-\epsilon\leq I_\text{e}\left(\text{X};\text{R}\right)\leq I\left(\text{X};\text{R}\right)+\epsilon\right]\geq1-\beta.
\end{align}
Substituting Eq.~(\ref{eq:eq119}) into Eq.~(\ref{eq:eq121}) delivers
\begin{align}
\label{eq:eq122}
\text{pr}\left[ I_\text{T}^\mathcal{F}-\epsilon\leq I_\text{e}\left(X;R\right)\leq I_\text{T}^\mathcal{A}+\epsilon\right]\geq1-\beta.
\end{align}
Thus, Algorithm~\ref{alg:CertificationofRQSSprotocols} accepts with probability at least $1-\beta$ if $X\in \mathcal{I}$.\\
Case (ii):
\begin{align}
\label{eq:eq123}
\text{pr}\left[ I_\text{e}\left(X;R\right)-\epsilon\leq I\left(X;R\right)\right]\geq1-\beta,
\end{align}
and
\begin{align}
\label{eq:eq124}
\text{pr}\left[I\left(X;R\right)\leq I_\text{e}\left(X;R\right)+\epsilon\right]\geq1-\beta.
\end{align}
Substituting Eq.~(\ref{eq:soundnessintermediate1})  and Eq.~(\ref{eq:soundnessintermediate2}) into Eq.~(\ref{eq:eq123}) and Eq.~(\ref{eq:eq124}), respectively,
delivers
\begin{align}
\label{eq:eq125}
\text{pr}\left[ I_\text{e}\left(X;R\right)\leq I_\text{T}^\mathcal{F}-\delta+\epsilon\right]\geq1-\beta,
\end{align}
and
\begin{align}
\label{eq:eq126}
\text{pr}\left[I_\text{T}^\mathcal{A}\leq I_\text{e}\left(X;R\right)-\delta+\epsilon\right]\geq1-\beta.
\end{align}
As~$\delta-\epsilon\geq \epsilon$, we conclude
\begin{align}
\label{eq:eq127}
\text{pr}\left[ I_\text{e}\left(X;R\right)\leq I_\text{T}^\mathcal{F}-\epsilon\right]\geq1-\beta,
\end{align}
and
\begin{align}
\label{eq:eq128}
\text{pr}\left[I_\text{T}^\mathcal{A}+\epsilon\leq I_\text{e}\left(X;R\right)\right]\geq1-\beta.
\end{align}
Thus, Algorithm~\ref{alg:CertificationofRQSSprotocols} rejects with probability at least $1-\beta$ if $X\notin \mathcal{I}$.
\end{proof}
\section{DISCUSSION}
\label{sec:Discussion} 
In this section we discuss our results. We have two main results.
The first result is a security analysis,
which assigns subsets of players to each of the three structures, namely, authorized, intermediate, and forbidden structures.
The second result is certification,
which is performed by a referee. In our security analysis, we not only determine structures for subset of players,
but we also quantify information leakage.
For certification we introduce a referee who has limited resources 
such as finite local oscillator field.
We now discuss these two results.

We base our approach on TRS03,
which divides subsets of players into authorized and forbidden structures.
TRS03 do not consider the intermediate structure because their security analysis is based on assuming infinite squeezing, but finite squeezing is responsible for information leakage, which leads us to introduce the intermediate structure based on ramp secret sharing concepts. Ramp quantum secret sharing has been considered before in two cases: discrete-variable threshold ramp quantum secret sharing~\cite{PhysRevA.72.032318} and entanglement sharing~\cite{ChoiRanHee2013Espv}. These analysis did not treat the continuous-variable case, however. In our case, for any amount of finite squeezing, we construct encoding and decoding procedures and thereby assign each subset to the correct structure.

Now we describe our result for certification. In our protocol, the dealer supplies the players with the encoded state, and in fact the state would be entangled with another share that goes directly to the referee. The referee identifies which subset of players are to transmit the decoded state to the referee, and the referee can combine this state with any shares that did not go through the players and then performs homodyne detection~\cite{0305-4470-37-29-010,aolita2015reliable}. The referee performs homodyne measurement, and,
if the local oscillator strength is infinite,
then standard homodyne theory suffices to describe the statistics.
We study the particular case of the referee performing tests based on Gaussian states and repeated measurements 
to allow the referee to estimate accurately the mean and covariance of the resultant state.
The referee's procedure is valid even in the case of limited local-oscillator strength. 

As our procedure is rather complicated 
and involves multiple parties,
we have augmented our analysis by including pseudocode to explain step-by-step instructions on how to complete the procedure. 
Our pseudocode analysis makes clear exactly what is required of each party in the procedure.
This pseudocode description could be a useful approach for describing future continuous-variable quantum-information protocols.

\section{Conclusions}
We have developed continuous-variable quantum mutual information with an external reference system in order to quantify the leakage of information and evaluate the security of continuous-variable quantum secret sharing protocols.
Furthermore, we prove that information leakage arising in the TRS03 scheme
monotonically decreases with reduced squeezing.
In addition, we introduce a certification process for continuous-variable quantum secret sharing in the framework of quantum-interactive proofs and ramp quantum secret sharing schemes.

Pseudocodes have been introduced in order to represent clearly the sequence of steps taken to solve the certification problem.
Subsequently, we provide a practical realization of the certification test using homodyne detection, including a sufficiency condition on the number of experimental runs the referee has to perform.
We prove that the statistical error in the referee’s estimated quantum mutual information scales with the inverse square root of number of rounds.

Our certification procedure assumes the extracted secret states are iid. In reality, this iid property does not hold due to the environmental noises.
Furthermore, in quantum secret sharing schemes, malicious parties might generate highly complicated entanglement among samples to fool the referee.
As a future line of research, it is important to extend our certification procedure to the case of samples
that are not independent and identically distributed.

Another useful avenue of research would be to analyze the effect of systematic errors in the referee's measurement procedure. As a final remark, we emphasize that our certification approach is applicable to certifying other quantum-information protocols such as summoning of quantum information in space time, quantum error correcting codes and quantum teleportation in the framework of quantum-interactive proof systems.

\appendix
\section{Calculation of quantum mutual information}
\label{sec:Appendix} 
The total density operator $\hat{\rho}_T$ of all shares and the reference system after the extraction procedure is
\begin{align}
	\hat{\rho}_\text{T}
		=&\frac{1}{\pi}\int_{\mathbb{R}^{2n+2}}\text{d}^n\bm{x}\text{d}^n\bm{x}'\,\text{d}y\text{d}y' \,\rho\left(y,x_1,y',x_1'\right)\ket{ y}\bra{ y'}\otimes\bigotimes_{\substack{i=1}}^{n}\ket{\xi_i}\bra{\xi_i'}\nonumber \\ 
& \times \exp\left\{ -\sum_{i=1}^{k-1}\left[\frac{y_i^2+y_i'^2}{2a^2}+\frac{a^2\left(z_i^2+z_i'^2\right)}{2}\right]\right\},
\end{align}
where
\begin{align}
\rho\left(y,x_1,y',x_1'\right)=&\exp\bigg[-\frac{\text{e}^{-2|\zeta|}\left(x_1+y\right)^2}{4}-\frac{\text{e}^{2|\zeta|}\left(y-x_1\right)^2}{4}-\frac{\text{e}^{-2|\zeta|}\left(x_1'+y'\right)^2}{4}\nonumber\\
&-\frac{\text{e}^{2|\zeta|}\left(y'-x_1'\right)^2}{4}\bigg].
\end{align}
The joint density operator
\begin{equation}
	\rho'=\bra{\omega' \eta'}\hat{\rho}\ket{\omega \eta},
\end{equation}
of the extracted secret and the reference system is obtained by tracing $\hat{\rho}_T$ over shares $\{2,3,\ldots n\}$. The resultant density matrix is
\begin{align}
\label{eq:position representation}
\rho'\left(\omega,\eta,\omega',\eta'\right)=&\frac{a}{\pi}\sqrt{\frac{1}{a^2+\frac{1}{2}\left(\text{e}^{2|\zeta|}+\text{e}^{-2|\zeta|}\right)v^2}}
\nonumber\\
&\times\exp\Bigg\{\left(-\frac{\text{e}^{-2|\zeta|}}{4}-\frac{\text{e}^{2|\zeta|}}{4}+\frac{\left(-4\text{e}^{-2|\zeta|}+4\text{e}^{2|\zeta|}\right)^2}{4a^2+\frac{\text{e}^{2|\zeta|}}{2}+\frac{\text{e}^{-2|\zeta|}}{2}v^2}\right)\left(\omega^2+\omega'^2\right)
\nonumber\\
&
+\left(\frac{\text{e}^{-4 |\zeta|}+\text{e}^{4 |\zeta|}+2}{16 a^2+8 \left(\text{e}^{2|\zeta|}+\text{e}^{-2|\zeta|}\right) v ^2}-\frac{\text{e}^{-2 |\zeta|}}{4} -\frac{\text{e}^{2 |\zeta|}}{4}\right)\left(\eta^2+\eta'^2
\right)
\nonumber\\
&+\bigg(\frac{\left(\text{e}^{2|\zeta|}+\text{e}^{-2|\zeta|}\right)}{2a^2+2\left(\text{e}^{2|\zeta|}+\text{e}^{-2|\zeta|}\right)v^2}-\frac{\text{e}^{-2|\zeta|}}{2}-\frac{\text{e}^{2|\zeta|}}{2}\bigg)\left(\omega\eta+\omega'\eta'\right)
\nonumber\\
&+\left(\frac{\left(\text{e}^{2|\zeta|}+\text{e}^{-2|\zeta|}\right)}{2a^2+2\left(\text{e}^{2|\zeta|}+\text{e}^{-2|\zeta|}\right)v^2}\right)\left(\eta\omega'+\eta'\omega\right)
\nonumber\\
&+\left(\frac{\text{e}^{-4 |\zeta|} \left(\text{e}^{4 |\zeta|}-1\right)^2}{8 a^2+4 \left(\text{e}^{2|\zeta|}+\text{e}^{-2|\zeta|}\right) v ^2}\right)\omega\omega'
\nonumber\\
&+\left(\frac{\text{e}^{-4 |\zeta|}+\text{e}^{4 |\zeta|}+2}{16 a^2+8 \left(\text{e}^{2|\zeta|}+\text{e}^{-2|\zeta|}\right) v ^2}\right)\eta\eta' \Bigg\},
\end{align}
where~$v^2=\bm{\gamma}_1\odot\bm{\gamma}_1$ for which~$\bm{\gamma}_1=\left(\gamma_{11},\gamma_{12},\ldots,\gamma_{1k-1}\right)$~(\ref{eq:encdectrs032}). Also,~$u^2=\bm{u}\odot\bm{u}$ where $\{\bm{u}_i\}$ are the coefficients of the expansion~$\alpha_j=\sum_{i=2}^{k-1} \bm{u}_i\beta_{ij}$ for which~$j=2,...,k-1.$ Then, by employing Eqs.~(\ref{eq:eqsix}),(\ref{eq:eqseven}), and (\ref{eq:eqeight}), we transform this density matrix into a Wigner function representation~(\ref{eq:eqsix}), namely
\begin{align}
W\left(q_1,p_1,q_2,p_2\right)=&\underbrace{\frac{2 a}{\pi^{2}} \sqrt{\frac{\text{e}^{2 |\zeta|}}{2 a^2 \text{e}^{2 |\zeta|}+\text{e}^{4 |\zeta|}+1}} \sqrt{\frac{a^2 \text{e}^{2 |\zeta|}}{2 a^2 \text{e}^{2 |\zeta|}+u^2 \left(\text{e}^{4 |\zeta|}+1\right)}}}_{N}
\nonumber\\
&\times\exp\bigg\{\underbrace{\left(-\frac{a^2 \left(\text{e}^{4 |\zeta|}+1\right)+2 \text{e}^{2 |\zeta|}}{2 a^2 \text{e}^{2 |\zeta|}+\text{e}^{4 |\zeta|}+1}\right)}_{\beta_1}q_1^2
\nonumber\\
&+\underbrace{\left(-\frac{a^2 \left(\text{e}^{4 |\zeta|}+1\right)}{2 a^2 \text{e}^{2 |\zeta|}+\text{e}^{4 |\zeta|}+1}\right)}_{\beta_2}q_2^2+\underbrace{\left(\frac{2 a^2 \left(\text{e}^{4 |\zeta|}-1\right)}{2 a^2 \text{e}^{2 |\zeta|}+\text{e}^{4 |\zeta|}+1}\right)}_{\beta_3}q_1q_2
\nonumber\\
&+\underbrace{\left(-\frac{a^2 \left(\text{e}^{4 |\zeta|}+1\right)+2 u^2 \text{e}^{2 |\zeta|}}{2 a^2 \text{e}^{2 |\zeta|}+u^2 \left(\text{e}^{4 |\zeta|}+1\right)}\right)}_{\gamma_1}p_1^2
\nonumber\\
&+\underbrace{\bigg(-\frac{a^2 \left(\text{e}^{4 |\zeta|}+1\right)}{2 a^2 \text{e}^{2 |\zeta|}+u^2 \left(\text{e}^{4 |\zeta|}+1\right)}\bigg)}_{\gamma_2}p_2^2
\nonumber\\
&+\underbrace{\left(-\frac{2 a^2 \left(\text{e}^{4 |\zeta|}-1\right)}{2 a^2 \text{e}^{2 |\zeta|}+u^2 \left(\text{e}^{4 |\zeta|}+1\right)}\right)}_{\gamma_3}p_1p_2
\bigg\}.
\end{align}
By using Eq.~(\ref{eq:eqnine}),
this Wigner function  is employed to derive the generic elements of the covariance matrix $\bm{V}$ corresponding to the joint reference and extracted-secret state.
The elements of $\bm{V}$ are
\begin{subequations}
\begin{align}
V_{12}&=V_{21}=V_{14}=V_{41}=V_{23}=V_{32}=V_{34}=V_{43}=0,\\
V_{11}&=N\frac{2\pi^2}{\beta_2^{1/2}\left(\beta_1-\frac{\beta_3^2}{4\beta_2}\right)^{3/2}\left(\gamma_1\gamma_2-\frac{\gamma_3^2}{4}\right)^{1/2}}, \\
V_{13}&=N\frac{\pi^2\beta_3}{2\left(\beta_1\beta_2-\frac{\beta_3^2}{4}\right)^{3/2}\left(\gamma_1\gamma_2-\frac{\gamma_3^2}{4}\right)^{1/2}}=V_{31},\\
V_{22}&=N\frac{2\pi^2}{\gamma_2^{1/2}\left(\gamma_1-\frac{\gamma_3^2}{4\gamma_2}\right)^{3/2}\left(\beta_1\beta_2-\frac{\beta_3^2}{4}\right)^{1/2}},\\
V_{24}&=N\frac{\pi^2\gamma_3}{2\left(\gamma_1\gamma_2-\frac{\gamma_3^2}{4}\right)^{3/2}\left(\beta_1\beta_2-\frac{\beta_3^2}{4}\right)^{1/2}}=V_{42},\\
V_{33}&=N\frac{2\pi^2}{\beta_1^{1/2}\left(\beta_2-\frac{\beta_3^2}{4\beta_1}\right)^{3/2}\left(\gamma_1\gamma_2-\frac{\gamma_3^2}{4}\right)^{1/2}},\\
V_{44}&=N\frac{2\pi^2}{\gamma_1^{1/2}\left(\gamma_2-\frac{\gamma_3^2}{4\gamma_1}\right)^{3/2}\left(\beta_1\beta_2-\frac{\beta_3^2}{4}\right)^{1/2}}.\\
\end{align}
\end{subequations}
The covariance matrix of the extracted secret and reference system denoted by~$V_\text{S}$ and~$V_\text{R}$ are 
\begin{equation}
\label{eq:vrs}
\bm{V}_\text{S}=\begin{pmatrix}
V_{11}&V_{12}\\
V_{21}&V_{22}
\end{pmatrix},\qquad
\bm{V}_\text{R}=\begin{pmatrix}
V_{33}&V_{34}\\
V_{43}&V_{44}
\end{pmatrix}.
\end{equation}
Also the joint covariance matrix of the extracted secret  and reference system is 
\begin{equation}
\bm{V}_{\rho^{\text{RS}}}=\begin{pmatrix}
V_{ij}
\end{pmatrix}.
\end{equation}
For convenience, let us also define
\begin{equation}
	\bm{C}:=\begin{pmatrix}
V_{13}&V_{14}\\
V_{23}&V_{24}
\end{pmatrix}.
\end{equation}
Using Eq.~(\ref{eq:eqtenfour}), symplectic eigenvalues of $\bm{V}_\text{S}$ and $\bm{V}_\text{R}$ denoted by $\nu_\text{S}$ and $\nu_\text{R}$ are
\begin{equation}
\nu_\text{R}=\sqrt{\det{\bm{V}_\text{R}}},\qquad\nu_\text{S}=\sqrt{\det{\bm{V}_\text{S}}},
\end{equation}
for which~$\bm{V}_\text{S}$ and~$\bm{V}_\text{R}$ are defined in Eq.~(\ref{eq:vrs}).
Symplectic eigenvalues of $\bm{V}_{\rho^\text{RS}}$ denoted by $\nu_{\pm}$ is calculated using Eq.~(\ref{eq:eqonesix}), therefore,
\begin{equation}\label{eq:2symplecticeigenvalues}
\nu_\pm=\sqrt{\frac{\Delta\pm\sqrt{\Delta^2-4\det \bm{V}_{\rho^\text{RS}}}}{2}},
\end{equation}
where $\Delta=\det{\bm{V}_\text{S}}+\det{\bm{V}_\text{R}}+2\det{\bm{C}}$.
\bibliography{ref}
\end{document}